\documentclass[11pt,letterpaper]{article}
\usepackage{times}
\usepackage[left=1in, right=1in, top=1in, bottom=1in]{geometry}
\usepackage{amsmath,amsthm,amssymb,paralist}
\usepackage{graphicx}
\usepackage{color}
\usepackage{wrapfig}
\usepackage[small,bf]{caption}
\usepackage{subfig}
\usepackage{paralist}
\usepackage{subfig}
\usepackage[small,compact]{title sec}
\usepackage{listings}
\usepackage{xcolor}
\usepackage{textcomp}
\usepackage{placeins}
\usepackage[ruled,noline,noend]{algorithm2e}
\usepackage[normalem]{ulem}
\usepackage{newfloat}
\DeclareFloatingEnvironment[name={Supplementary Figure}]{suppfigure}

\usepackage{tikz}
\usetikzlibrary{positioning}

\renewcommand{\Pr}{\mathbf{Pr}}
\newcommand{\E}{\mathbf{E}}

\newcommand{\G}{\mathcal{G}}
\newcommand{\I}{\mathcal{I}}
\newcommand{\C}{\mathcal{C}}
\newcommand{\Pathway}{\mathcal{D}}
\newcommand{\Set}{\mathcal{S}}
\newcommand{\Patients}{\mathcal{P}}

\newcommand{\BO}[1]{{O}\left(#1\right)}

\newcommand{\BOM}[1]{{\Omega}\left(#1\right)}
\newtheorem{theorem}{Theorem}
\newtheorem{proposition}{Proposition}

\DeclareMathOperator*{\argmax}{arg\,max}

\newcommand{\problemname}{max $k$-set log-rank}
\newcommand{\problemgraphname}{max connected $k$-set log-rank}
\newcommand{\algname}{NoMAS}
\newcommand{\opt}{$OPT$}

\begin{document}
\title{Finding Mutated Subnetworks Associated with Survival in Cancer}
\author{Tommy Hansen\thanks{Department of Mathematics and Computer Science, University of Southern Denmark, Odense (Denmark).}\\\texttt{tvhansen33@gmail.dk}\and Fabio Vandin \thanks{Department of Information Engineering, University of Padova, Padova (Italy).} $^{,*,}$\thanks{Department of Computer Science, Brown University, Providence, RI (USA).} $^,$\thanks{Corresponding author.}\\\texttt{vandinfa@dei.unipd.it}}
\date{}
\maketitle{}

\begin{abstract}

Next-generation sequencing technologies allow the measurement of somatic mutations in a large number of patients from the same cancer type.
One of the main goals in the analysis of these mutations is the identification of mutations associated with clinical parameters, for example survival time. 
The identification of mutations associated with survival time is hindered by the genetic heterogeneity of mutations in cancer, due to the fact that genes and mutations act in the context of \emph{pathways}. To identify mutations associated with survival time it is therefore crucial to study mutations in the context of interaction networks.

In this work we study the problem of identifying subnetworks of a large gene-gene interaction network that have mutations associated with survival. We formally define the associated computational problem by using a score for subnetworks based on the test statistic of the log-rank test, a widely used statistical test for comparing the survival of two given populations. We show that the computational problem is NP-hard and we propose a novel algorithm, called \underline{N}etwork \underline{o}f \underline{M}utations \underline{A}ssociated with \underline{S}urvival (\algname), to solve it. NoMAS is based on the color-coding technique, that has been previously used in other applications to find the highest scoring subnetwork with high probability when the subnetwork score is additive. In our case the score is not additive; nonetheless, we prove that under a reasonable model for mutations in cancer \algname\ does identify the optimal solution with high probability. We test \algname\ on simulated and cancer data, comparing it to approaches based on single gene tests and to various greedy approaches. We show that our method does indeed find the optimal solution and performs better than the other approaches. Moreover, on two cancer datasets our method identifies subnetworks with significant association to survival when none of the genes has significant association with survival when considered in isolation.

\end{abstract}

\newpage
\section{Introduction}
Recent advances in next-generation sequencing technologies have enabled the collection of sequence information 
from many genomes and exomes, with many large human and cancer genetic studies measuring mutations in all genes for a large number of patients of a specific disease \cite{Cancer-Genome-Atlas-Research-Network:2014aa,Cancer-Genome-Atlas-Research-Network:2013aa,Cancer-Genome-Atlas-Network:2015aa}. One of the main challenges in these studies is the interpretation of such mutations, in particular the identification of mutations that are clinically relevant. For example, in large cancer studies one is interested in finding somatic mutations that are associated with survival and that can be used for prognosis and therapeutic decisions.
One of the main obstacles in finding mutations that are clinically relevant is the large number of mutations present in each cancer genome. Recent studies have shown that each cancer genome harbors hundreds or thousands of somatic mutations~\cite{Garraway:2013uq}, with only a small number (e.g., $\le 10$) of \emph{driver} mutations related to the disease, while the vast majority of mutations are \emph{passenger}, random mutations that are accumulated during the process that leads to cancer but not related to the disease~\cite{Vogelstein:2013fk}.

In  recent years, several computational and statistical methods have been designed to identify driver mutations and distinguish them from passenger mutations using mutation data from large cancer studies~\cite{Raphael:2014aa}. Many of these methods analyze each gene in isolation, and use different single gene scores (e.g., mutation frequency, clustering of mutations, etc.) to identify significant genes~\cite{Dees:2012aa,Lawrence:2013aa,Tamborero:2013aa}. While useful in finding driver genes, these methods suffer from the extensive \emph{heterogeneity} of mutations in cancer, with different patients showing mutations in different cancer genes~\cite{Kandoth:2013aa}. One of the reasons of such mutational heterogeneity is the fact that driver mutations do not target single genes but rather \emph{pathways}~\cite{Vogelstein:2013fk}, groups of interacting genes that perform different functions in the cell. Several methods have been recently proposed to identify significant groups of interacting genes in cancer~\cite{Hofree:2013aa,Leiserson:2015aa,Vandin:2012ab,Leiserson:2015ab,shrestha2014hit,kim2015memcover}. Many of these methods integrate mutations with interactions from genome-scale interaction networks, without restricting to already known pathways, that would hinder the ability to discover new important groups of genes.

In addition to mutation data, large cancer studies often collect also clinical data, including survival information, regarding the patients. An important feature of survival data is that it often contains \emph{censored} measurements~\cite{Kalbfleisch:2002fk}: in many studies a patient may be alive at the end of the study or may leave the study before it ends, therefore only a lower bound to the survival of the patient is known. Survival information is crucial in identifying mutations that have a clinical impact. However, this information is commonly used only to assess the significance of candidate genes or groups of genes identified using other computational methods~\cite{Hofree:2013aa,Cancer-Genome-Atlas-Research-Network:2011aa}, as the ones described above, and there is a lack of methods that integrate mutations, interaction information, and survival data to directly identify groups of interacting genes associated with survival.

The field of survival analysis has produced an extensive literature on the analysis of survival data, in particular for the comparison of the survival of two given populations (sets of samples)~\cite{Kalbfleisch:2002fk}. The most commonly used test for this purpose is the log-rank test~\cite{citeulike:7445010,pmid5910392}. 
In genomic studies we are not given two populations, but a single set of samples, and are required to identify  mutations that are associated with survival. The log-rank test can be used to this end to identify single genes associated with survival time by comparing the survival of the patients with a mutation in the gene with the survival of the patients with no mutation in the gene. The other commonly used test, the Cox Proportional-Hazards model~\cite{Kalbfleisch:2002fk}, is equivalent to the log-rank test when the association of a binary feature with survival is tested, as it is in the case of interest to genomic studies.
For a given group of genes, one can \emph{assess} the association of mutations in the genes of the group with survival by comparing the survival of the patients having a mutation in at least one of the genes with the survival of the patients with no mutation in the genes. However, this approach cannot be used to \emph{discover} sets of genes, since one would have to screen all possible subsets of genes and test their association with survival, and the number of subsets of genes to screen is enormous even considering only groups of genes interacting in a protein interaction network (e.g., there are $> 10^{15}$ groups of $8$ interacting genes in HINT+HI2012 network~\cite{Leiserson:2015ab}).

Color-coding is a probabilistic method that was originally described for finding simple paths, cycles and other small subnetworks of size $k$ within a given network \cite{alon1994color}. The core of the color-coding technique is the assignment of random colors to the vertices, as a result of which the search space can be reduced, by restricting the subnetworks under consideration to \textit{colorful} ones, those in which each vertex has a distinct color. For the identification of colorful subnetworks, dynamic programming is employed. The process is repeated until the desired subnetwork has been identified, that is having been colorful at least once, with high probability. When the dynamic programming algorithm is polynomial in $n$ and the subnetworks being screened are of size $k \in O(\log n)$, the overall running time of the color-coding method too remains polynomial in $n$.

 \subsection{Our Contribution}
 
 In this paper we study the problem of finding sets of interacting genes with mutations associated to survival using 
 data from large cancer sequencing studies and interaction information from a genome-scale interaction network. We focus on the widely used log-rank statistic as a measure of the association between mutations in a group of genes and survival.
 Our contribution is threefold: first, we formally define the problem of finding the set of $k$ genes whose mutations show the maximum association to survival time by using the log-rank statistic as a score for a set of genes, and we show that such problem is NP-hard. We show that the problem remains hard when the set of $k$ genes is required to form a connected subnetwork in a large graph with at least one node of large degree (\emph{hub}).
 
 Second, we propose an efficient algorithm, \underline{N}etwork \underline{o}f \underline{M}utations \underline{A}ssociated with \underline{S}urvival (\algname), based on the color-coding technique, to identify subnetworks associated with survival time. Color-coding has been previously used to find high scoring graphs for bioinformatics applications~\cite{dao2011optimally,Hormozdiari:2015aa} when the score for a subnetwork is \emph{set additive} (i.e., the score of a subnetwork is the sum of the scores of the genes in the subnetwork). In our case the log-rank statistic is not set additive, and we prove that there is a family of instances for which our algorithm cannot identify the optimal solution. Nonetheless, we prove that under a reasonable model for mutations in cancer our algorithm identifies the optimal solution with high probability. 
 
Third, we test our algorithm on simulated data and  on data from three large cancer studies from The Cancer Genome Atlas (TCGA). On simulated data, we show that our algorithm does find the optimal 
solution while being much more efficient than the exhaustive algorithm that screens all sets of genes. On cancer data, we show that our algorithm finds the optimal solution for all values of $k$ for which the use of the exhaustive algorithm is feasible, and identifies better solutions (in terms of association to survival) than a greedy algorithm similar to the one used in~~\cite{Reimand:2013aa}.
Moreover, we show that \algname\ identifies better solutions than using an (additive) score (i.e., the same gene score used in~\cite{Vandin:2012aa}) for a set of genes.
For the cancer datasets, we show that our algorithm identifies novel groups of genes associated with survival where none of the genes is associated with survival when considered in isolation.
 
  \subsection{Related Work}

Few methods have been proposed to identify groups of genes with mutations associated with survival in genomic studies. The work of~\cite{Vandin:2012aa} combines mutation and survival data with interaction information using a diffusion process on graphs starting from gene scores derived from $p$-values of individual genes, but did not consider the problem of directly identifying groups of genes associated with survival. The work of~\cite{Reimand:2013aa} combines mutation information and patient survival to identify subnetworks of a kinase-substrate interaction network associated with survival. It only focuses on phosphorylation-associated mutations, and the approach is based on a local search algorithm that builds a subnetwork by starting from one seed vertex and then greedily adding neighbours (at distance at most 2) from the seed, extending the approach used in different types of network analyses~\cite{Chuang:2007aa}. A similar greedy approach is used by \cite{Wu:2012aa} to identify groups of genes significantly associated with survival in cancer from gene expression data. 
For gene expression studies, \cite{chowdhury2011subnetwork} proposes an approach to enumerate dysregulated subnetworks in cancer based on an efficient search space pruning strategy, inspired by previous work on the identification of association rules in databases~\cite{smyth1992information}. \cite{patel2013network} uses the general approach described in \cite{chowdhury2011subnetwork} to identify subnetworks of genes with
expression status associated to survival.

Color-coding has been previously used to count or search for subgraphs of large interaction networks (\cite{Alon:2008aa,Bruckner:2010aa}).  Color-coding has also been used to identify groups of interacting genes in an interaction network that are associated with a phenotype of interest, but restricted to additive scores for sets of genes (i.e., the score of a set is the sum of the scores of the single genes); for example, \cite{dao2011optimally} uses color-coding to find optimally discriminative subnetwork markers that predict response to chemotherapy from a large interaction network by defining a single gene score as $- \log_{10} d(g)$, where $d(g)$ is the discriminative score for gene $g$ (i.e., a measure of the ability of $g$ to discriminate two classes of patients); similarly,
\cite{Hormozdiari:2015aa} uses color-coding to find groups of interacting genes with discriminative mutations in case-control studies, using as gene score the $- \log_{10}$ of the $p$-value from the binomial test of recurrence of mutations in the cases (while limiting the number of mutations in the controls).
 
\section{Methods and Algorithms}

In this section we define the model and the algorithm used in this work. The remaining of the section is organized as follows: Section~\ref{sec:cp} introduces preliminary definitions, the computational problem and presents its computational complexity;  Section~\ref{sec:alg} introduces the algorithm we design to solve the problem; Section~\ref{sec:analysis} presents the analysis of the algorithm, including the analysis under a reasonable model for mutations in cancer. Due to space constraints, proofs are omitted; they will appear in the final version of this extended abstract. (Proof sketches for our results are given in Appendix.)

\subsection{Computational Problem}
\label{sec:cp}

In survival analysis, we are given two populations (i.e., sets of samples) $P_0$ and $P_1$, and for each sample $i \in P_0 \cup P_1$ we have its survival data:  \begin{inparaenum}[i)] \item the survival time $t_i$ and \item the censoring information $c_i$, \end{inparaenum} where $c_i=1$ if $t_i$ is the exact survival time for sample $i$ (i.e., sample $i$ is not censored), and $c_i=0$ if $t_i$ is a lower bound to the survival time for sample $i$ (i.e., sample $i$ is censored).   Let $m_0$ be the number of samples in $P_0$, $m_1$ be the number of samples in $P_1$, and $m=m_0+m_1$ be the  total  number of samples. Without loss of generality, the samples are $\{1,2,\dots,m\}$, the survival times are $t=1,2,\dots,m$, with $t_i=i$ (i.e., the samples are sorted by increasing values of survival), and we assume that there are no ties in survival times. The survival data is represented by two vectors $\mathbf{c}$ and $\mathbf{x}$, with $c_i$ representing the censoring information for sample $i$, and $x_i$ represents the population information: $x_i=1$ if sample $i$ is in population $P_1$, and $x_i=0$ otherwise. 
 Given the survival data for two populations $P_0$ and $P_1$, the significance in the difference of survival between $P_0$ and $P_1$  can be assessed by the widely used log-rank test~\cite{citeulike:7445010,pmid5910392}.
The log-rank statistic is 
\begin{equation}
\label{eq:nlr}
V(\mathbf{x},\mathbf{c}) =\sum_{j=1}^m c_j \left(x_j -  \frac{m_1 -\sum_{i=1}^{j-1} x_i}{m-j+1}\right)
\end{equation}

Under the (null) hypothesis of no difference in survival between $P_0$ and $P_1$, the log-rank statistic asymptotically follows a normal distribution $\mathcal{N}(0,\sigma^2)$, where the standard deviation\footnote{In the literature two different standard deviations (corresponding to two related but different null distributions, permutational and conditional) have been proposed for the normal approximation of the distribution of the log-rank statistic; we have previously shown~\cite{Vandin:2015aa} that the one we use here (corresponding to the permutational distribution) is more appropriate for genomic studies.} is given by:
$\sigma(\mathbf{x},\mathbf{c}) = \sqrt{\frac{m_0 m_1}{m(m-1)} \left( \left(\sum_{j=1}^m c_j\right)  - \sum_{j=1}^m c_i \frac{1}{m-j+1}  \right)}.$

Thus the normalized log-rank statistic, defined as $\frac{V(\mathbf{x},\mathbf{c})}{\sigma(\mathbf{x},\mathbf{c})}$, asymptotically follows a standard normal $\mathcal{N}(0,1)$ distribution, and the deviation of $\frac{V(\mathbf{x},\mathbf{c})}{\sigma(\mathbf{x},\mathbf{c})}$ from $0$ is a measure of the difference in survival between $P_0$ and $P_1$.

In genomic studies, we are given mutation data for a set $\G$ of $n$ genes in a set $\Patients$ of $m$ samples, represented by a mutation matrix $M$ with $M_{i,j}=1$ if gene $i$ is mutated in patient $j$ and $M_{i,j}=0$ otherwise. We are also given survival data (survival time and censoring information) for all the $m$ samples. Given a set $\Set \subset \G$ of genes, we can assess the association of mutations in the set $\Set$ with survival by comparing the survival of the population $P^{\Set}_1$ of samples with a mutation in at least one gene of $\Set$ and the survival of the population $P^{\Set}_0$ of samples with no mutation in the genes of $\Set$. That is, $P^{\Set}_0 = \{j \in \Patients: \sum_{i\in \Set} M_{i,j} = 0 \}$ and $P^{\Set}_1 = \{j \in \Patients: \sum_{i\in \Set} M_{i,j} > 0 \}$.

Given the set $\mathcal{G}$ of all genes for which mutations have been measured, we are interested in finding the set $\Set \subset \mathcal{G}$ with $|\Set| = k$ that 
has maximum association with survival  by finding the set $\Set$ that maximizes the absolute value of the normalized log-rank statistic. Given a set $\Set$ of genes, let $\mathbf{x}^{\Set}$ be a $0-1$ vector, with $x^{\Set}_i=1$ if at least one gene of $\Set$ is mutated in patient $i$, and $x^{\Set}_i=0$ otherwise. The normalized log-rank statistic for the set $\Set$ is then $\frac{V(\mathbf{x}^{\Set}, \mathbf{c})}{\sigma(\mathbf{x}^{\Set}, \mathbf{c})}$. 
Note that for a given set of patients the censoring information $\mathbf{c}$ is fixed, therefore we can consider the log-rank statistic as a function $V(\mathbf{x}^{\Set})$ of $\mathbf{x}^{\Set}$ only. Analogously, we can rewrite $\sigma(\mathbf{x}^{\Set},\mathbf{c}) =  \sigma(\mathbf{x}^{\Set}) f(\mathbf{c})$, where $ \sigma(\mathbf{x}^{\Set}) = \sqrt{m_1 (m -m_1) }$ with $m_1 = | P^{\Set}_1|$, and  $f(\mathbf{c}) =  \sqrt{\frac{1}{m(m-1)} \left( \left(\sum_{j=1}^m c_j\right)  - \sum_{j=1}^m c_j \frac{1}{m-j+1}  \right)}$ does not depend on  $\mathbf{x}^{S}$ and is fixed given $\textbf{c}$. 

To identify the set of $k$ genes most associated with survival, we can then consider the score $|w(\Set)| = \left|\frac{V(\mathbf{x}^{\Set})}{\sigma(\mathbf{x}^{\Set})}\right|$. For ease of exposition in what follows we consider the score $w(\Set)$, corresponding to a one tail log-rank test for the identification of sets of genes with mutations associated with reduced survival; the identification of sets of genes with mutations associated with increased survival is done in an analogous way by maximizing the score $-w(\Set)$.
We define the following problem.

\begin{problem}{The \problemname\ problem}
Given a set $\mathcal{G}$ of genes, an $n \times m$ mutation matrix $M$ and the survival information (time and censoring) for the $m$ patients in $M$, find the set ${\Set} \subset \mathcal{G}$ of $k$ genes maximizing $w({\Set})$.
\end{problem}

We have the following.

\begin{theorem}
\label{thm:nphard}
The \problemname\  problem is NP-hard.
\end{theorem}

We now define the \problemgraphname\ problem that is analogous to the \problemname\ problem but requires feasible solutions to be connected subnetworks of a given graph $\I$, representing gene-gene interactions.

\begin{problem}{The \problemgraphname\ problem}
Given a set $\G$ of genes, a graph $\I = (\G, E)$ with $E \subseteq \G \times \G$,  an $n \times m$ mutation matrix $M$ and the survival information (time and censoring) for he $m$ patients in $M$, find the set ${\Set}$ of $k$ genes maximizing $w({\Set})$ with the constraint that the subnetwork induced by ${\Set}$ in $\I$ is connected.
\end{problem}

If $\I$ is the complete graph, the \problemgraphname\ problem is the same as the \problemname\ problem. Thus, the \problemgraphname\ problem is NP-hard for a general graph. However, we can prove that the problem is NP-hard for a much more general class of graphs.

\begin{theorem}
The \problemgraphname\ problem on graphs with at least one node of degree $\BO{n^\frac{1}{c}}$, where $c>1$ is constant, is NP-hard.
\end{theorem}

\subsection{Algorithm}
\label{sec:alg}

We design a new algorithm, \underline{N}etwork \underline{o}f \underline{M}utations \underline{A}ssociated with \underline{S}urvival (\algname), to solve the \problemgraphname\ problem. The algorithm 
is based on an adaptation of  the color-coding technique~\cite{alon1994color}. Our algorithm is analogous to other color-coding based algorithms that have been used before to identify subnetworks associated with phenotypes in other applications where the score is additive~\cite{dao2011optimally,Hormozdiari:2015aa}.

The input to \algname\ is an undirected graph $G=(V,E)$, an $n \times m$ mutation matrix $M$, and the survival information $\mathbf{x},\mathbf{c}$ for the $m$ patients in $M$.  \algname\ first identifies a subnetwork $\Set$ with high score $\frac{w(\mathbf{x}^{\Set})}{\sigma(\mathbf{x}^{\Set})}$, and then uses a permutation test to assess the significance of the subnetwork. 

To identify a subnetwork of high weight, the algorithm proceeds in iterations. In each iteration \algname\ colors $G$ with $k$ colors by assigning to each vertex $v$ a color $\mathcal{C}(v) \in \{1,\dots,k\}$ chosen uniformly at random. 
For a given coloring of $G$, a subnetwork $\mathcal{S}$ is said to be \textit{colorful} if all vertices in $\mathcal{S}$ have distinct colors. The  \textit{colorset} of $\mathcal{S}$ is the set of colors of the vertices in $\mathcal{S}$ . Note that the number of different colorsets (subsets of $\{1,\dots,k\}$) is $2^k$.
In each iteration the algorithm efficiently identifies high-scoring colorful subnetworks, and at the end the highest-scoring subnetwork among all iterations is reported.

Consider a given coloring of $G$. Let $W$ be a $(2^k-1) \times |V|$ table with a row for each non-empty colorset and a column for each vertex in $G$. Entry $W(T, u)$ stores the set of vertices of one connected colorful subnetwork that has colorset $T$ and includes vertex $u$. 
Entries of $W$ can be filled by dynamic programming. For colorsets of size $1$, the corresponding rows in $W$ are filled out trivially: $W(\{\alpha\}, u) = \{u\}$ if $\alpha = \mathcal{C}(u)$, and $W(\{\alpha\}, u) = \emptyset$ otherwise.

\noindent For entry $W(T, u)$ with $|T| \geq 2$, \algname\ computes $W(T, u)$ by combining a previously computed $W(Q, u)$ for $u$ with another previously computed $W(R, v)$ where $v$ is  a neighbor of $u$ in $G$, ensuring that the resulting subnetwork is connected and contains $u$. Colorfulnes is ensured by selecting $Q$ and $R$ such that $Q \cap R = \emptyset$ and $Q \cup R = T$, and in turn ensures that $W(T, u)$ contains $|T|$ distinct vertices.
Note that for a given $T$ the choice of $Q$ uniquely defines $R$. Thus, for each neighbor $v$ of $u$ there are (at most) $2^{|T|-1}$ possible combinations.
Let $\mathcal{S}^\prime(T,u)$ be the set of all colorful subnetworks that can be obtained by combining an entry $W(Q,u)$ for  $u$ and an appropriate entry $W(R,v)$ for a neighbor $v$ of $u$ so that $Q \cup R = T, Q \cap R = \emptyset$. That is:

$\mathcal{S}^\prime(T,u) = \bigcup_{\substack{v:(u,v) \in E\\ Q\cup R = T, Q \cap T = \emptyset}} \left\{ W(Q,u) \cup W(R,v) \right\}.$

(In the definition of $\mathcal{S}^\prime(T,u)$ we assume that the union with $\emptyset$ returns $\emptyset$.)
$W(T,u)$ stores the element of $\mathcal{S}^\prime(T,u) $ with largest value of our objective function, that is
$W(T, u) = \argmax_{\mathcal{S} \in \mathcal{S}^\prime(T, u)} w(\mathcal{S}).$
At the end, the best solution is identified by finding the entry of $W$ of maximum weight. (See Appendix for pseudo code and illustrations of \algname).

After identifying the best solution $\mathcal{S}$ for the mutation matrix $M$, \algname\ assesses its statistical significance by \begin{inparaenum}[i)] \item estimating the p-value $p(\mathcal{S})$ for the log-rank statistic (using a Monte-Carlo estimate with $10^8$ samples), and then \item using a permutation test in which $\mathcal{S}$ is compared to the best solution $\mathcal{S}^p$ for the mutation matrix $M^p$ obtained by randomly permuting the rows of $M$. \end{inparaenum} A total of $100$ permutations are performed and the \textit{permutation} p-value is recorded as the ratio of permutations in which $w({S}^p) \geq w(\mathcal{S})$.
While the $p$-value from the log-rank test reflects the association between mutations in the subnetwork and survival, the permutation $p$-value assesses whether a subnetwork with association with survival at least as extreme as the one observed in the input data can be observed when the genes are placed randomly in the network.
Note that we can identify multiple solutions by considering different entries of $W$ (even if the same solution may appear in multiple entries of $W$), and we obtain a permutation $p$-value for the $i$-th top scoring solution by comparing its score with the score of the $i$-th top scoring solution in the permuted datasets.
 Analogously, \algname\ identifies sets that minimize $w(\Set)$ (sets associated to increased survival) by maximizing the score $-w(\Set)$.

\paragraph{Parallelization.} The computation of $W$ is parallelized using $N \leq |V|$ processors. All entries of $W$ are kept in shared memory and  $|V|/N$ unique columns uniformly at random are assigned to each processor.

Entries of $W$ are computed in order of increasing colorset sizes. We define the $i$-th \emph{colorset group} as the set of all $\binom{k}{i}$ colorsets of size $i$. We exploit the fact that rows within the $i$-th colorset group are computed by reading entries exclusively from rows belonging to colorset groups $< i$. 
When a processor has finished the rows of the $i$-th colorset group it waits for the other processors to do the same. When the last processor completes the $i$-th colorset group, all $N$ processors can safely begin to compute rows of colorset group $i+1$.
In total, $k$ synchronization steps are carried out, one for each colorset group.

\subsection{Analysis of \algname}
\label{sec:analysis}

We consider the performance of \algname\ excluding the permutation test.
The log-rank statistic $w(\mathcal{S})$ is computed in time $\BO{m_1} \in \BO{m}$.  The total time complexity for computing a single entry $W(T, u)$ is then bounded by $\BO{m\text{deg}(u)2^{|T|-1}} \in \BO{m\text{deg}(u)2^k}$, where $deg(u)$ is the degree of $u$ in $G$. Given a coloring of $G$, the computation of the entire table can thus be performed in time $\BO{2^k \sum_{u\in V} m \text{deg}(u) 2^{k}} \in \BO{m|E|4^k}$. If $L$ iterations are performed, then the complexity of the algorithm is $\BO{Lm|E|4^k}$.

Let \opt\ be the optimal solution. If the score $w(\Set)$ was set additive, as the scores considered in previous applications of color-coding for optimization problems on graphs, to discover \opt\ it would be sufficient that \opt\ be colorful, that happens with probability $k!/k^k \geq e^{-k}$ for each random coloring. Therefore $\BO{\ln(1/\delta)e^k}$ iterations would be enough to ensure that the probability of \opt\ not being discovered is $\le \delta$, resulting in an overall time complexity of $\BO{m\ln(1/\delta)|E|(4e)^k}$.

However, our score $w(\Set)$ is not set additive (e.g., if two genes in $\Set$ have a mutation in the same patient the weight of the patient is considered only once in $w(\Set)$). Therefore, while \opt\ being colorful is still  a necessary condition for the algorithm to identify \opt, the colorfulness of \opt\ is not a sufficient condition.
In fact, we have the following.

\begin{proposition}
\label{thm:badfamily}
For every $k \geq 3$ there is a family of instances of the \problemgraphname\ problem and colorings for which \opt\ is not found by \algname\ when it is colorful.
\end{proposition}

Even more, we prove that when mutations are placed arbitrarily then for every subnetwork $\Set$ and a given coloring of $\Set$,  \emph{any} color-coding algorithm that adds subnetworks of size $k$ to $W$ by merging neighboring subnetworks of size $< k$ could be ``fooled'' to not add $\Set$ to $W$ by simply adding $3$ vertices to $G$ and assigning them a specific color.

\begin{theorem}
\label{thm:colorbad}
For any optimal colorful connected subnetwork $\mathcal{S}$ of size $k \geq 3$ and any color-coding algorithm $\mathcal{A}$ which obtains subnetworks with colorsets of cardinality $i$ by combining $2$ subnetworks with colorsets of cardinality $< i$, by adding $3$ neighbors to $\mathcal{S}$ we have that $\mathcal{A}$ may not discover $S$ .
\end{theorem}

Intuitively, Proposition~\ref{thm:badfamily} and Theorem~\ref{thm:colorbad} show that if mutations are placed adversarially (and the optimal solution \opt\ has many neighbors), our algorithm may not identify \opt.
However, we prove that our algorithm identifies the optimal solution under a generative model for mutations, that we deem the \emph{Planted Subnetwork Model}. We consider $w(\Set)$ as the unnormalized version of the log-rank statistic. In this model:
\begin{inparaenum}[i)]
\item\label{cond1} there is a subnetwork $\Pathway$, $|\Pathway| = k$, with $w(\Pathway) \ge c m $, for a constant $c > 0$; 
\item\label{cond2} each gene $g \in \Pathway$ is such that $w(\Pathway) - w(\Pathway \setminus \{g\}) \ge \frac{c' m}{k}$, for a constant $c' > 0$;
\item\label{cond3} for each gene $g \in \Pathway$: $w(\{g\}) > 0$;
\item\label{cond4} for each gene $\hat{g} \notin \Pathway$, $\hat{g}$ is mutated with probability $p_g$ in each patient, independently of all other events (and of survival time and censoring status).
\end{inparaenum}

Intuitively, \ref{cond1}) above states that the subnetwork $\Pathway$ has mutations associated with survival; \ref{cond2}) states that each gene $g\in \Pathway$ contributes to the association of mutations in $\Pathway$ to survival; \ref{cond3}) states that each gene $g\in \Pathway$ should have the same association to survival (increased or decreased) as $\Pathway$; and \ref{cond4}) states that all mutations outside $\Pathway$ are independent of all other events (including survival time and censoring of patients).

We show that when enough samples are generated from the model above, our algorithm identifies the optimal solution with the same probability guarantee given by the color-coding technique for additive scores.

\begin{theorem}
Let $M$ be a mutation matrix corresponding to $m$ samples from the Planted Subnetwork Model. If $m \in \BOM{k^4(k+\varepsilon) \ln n}$ for a given constant $\varepsilon > 0$ and $\BO{\ln(1/\delta)e^k}$ color-coding iterations are performed, then our algorithm identifies the optimal solution $\Pathway$ to the \problemgraphname\ with probability $\ge 1 - \frac{1}{n^\varepsilon} - \delta$.
\end{theorem}

\section{Results}
We assessed the  performance of \algname\ by using simulated and cancer data. We compared \algname\ to the exhaustive algorithm that identifies the subnetwork of $k$ vertices with the highest score $w(S)$ for the values of $k$ for which we could run the exhaustive algorithm (we implemented a parallelized version of the algorithm described in~\cite{maxwell2014efficiently} to efficiently enumerate all connected subnetworks), to three variants of a greedy algorithm similar to the one from~\cite{Reimand:2013aa}, and to the use of a score given by the sum of single gene scores.  Cancer data is obtained from The Cancer Genome Atlas (TCGA). In particular, we consider somatic mutations (single nucleotide variants and small indels) for $268$ samples of glioblastoma multiforme (GBM), $315$ samples of ovarian adenocarcinoma (OV), and $174$ samples of lung squamous cell carcinoma (LUSC) for which survival data is available.

For all our experiments we used as interaction graph $G$ the graph derived from the application of a diffusion process on the HINT+HI2012 network\footnote{\texttt{http://compbio-research.cs.brown.edu/pancancer/hotnet2/}}, a combination of the HINT network~\cite{Das:2012aa} and the HI-2012~ \cite{Yu:2011aa}set of protein-protein interactions, previously used in~\cite{Leiserson:2015aa}. The details of the diffusion process are described in~\cite{Leiserson:2015aa}. In brief, for two genes $g_i,g_j$ the diffusion process gives the amount of heat $h(g_i,g_j)$ observed on $g_j$ when $g_i$ has one mutation, and the amount of heat $h(g_j,g_i)$ observed on $g_i$ when $g_j$ has one mutation. The graph used for our analyses is obtained retaining an edge between $g_i$ and $g_j$ if $\max\{h(g_i,g_j), h(g_j,g_i)\} \geq 0.012$. The resulting graph has $9859$ vertices and $42480$ edges, with the maximum degree of a node being $438$. In all our experiments we removed mutations in genes mutated in $< 3$ of the samples. For cancer data, this resulted in $780$ mutated genes in GBM, $890$ in OV, and $2915$ in LUSC.

The machine, on which all our experiments were carried out, consists of two CPUs of the type Intel Xeon E5-2698 v3 (2.30GHz), each with 16 physical cores, for a total of 64 virtual cores, and 16 banks of Samsung 32GB DDR4 (2133 MHz) memory modules for a total of 512GB of memory.

The remaining of the section is organized as follow: Section~\ref{sec:res_sim} presents the results on simulated data, while Section~\ref{sec:res_cancer} presents the results on cancer data. 

\subsection{Simulated data}
\label{sec:res_sim} 
We assess the performance of \algname\ on simulated data generated under the Planted subnetwork Model. The subnetwork $\Pathway \subset \G, |\Pathway| = k$ associated with survival is generated by a random walk on the graph $G$. We model the association of $\Pathway$ to survival by mutating with probability $p$ one gene of $\Pathway$ chosen uniformly at random in each sample among the $\frac{m}{4}$ of lowest survival.
All other genes in $\Pathway$ are mutated independently with probability $0.01$ in all samples, to simulate passenger mutations (not associated with survival) in $\Pathway$~\cite{Lawrence:2013aa}.
For genes in $\G \setminus \Pathway$, we used the same mutation frequencies observed in the GBM study, and mutate each gene independently of all other events.

\begin{figure}[ht]
\centering
\includegraphics[scale=0.40]{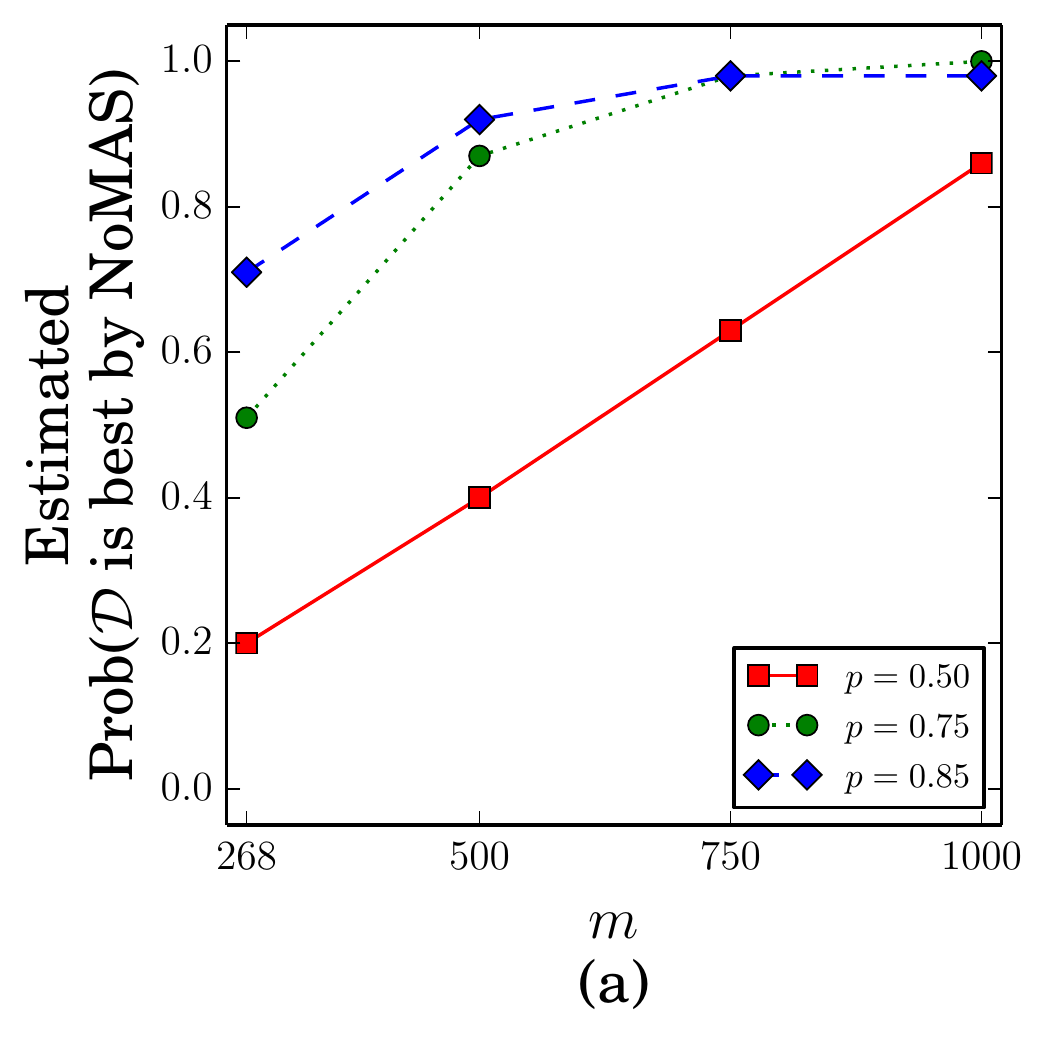}
\includegraphics[scale=0.40]{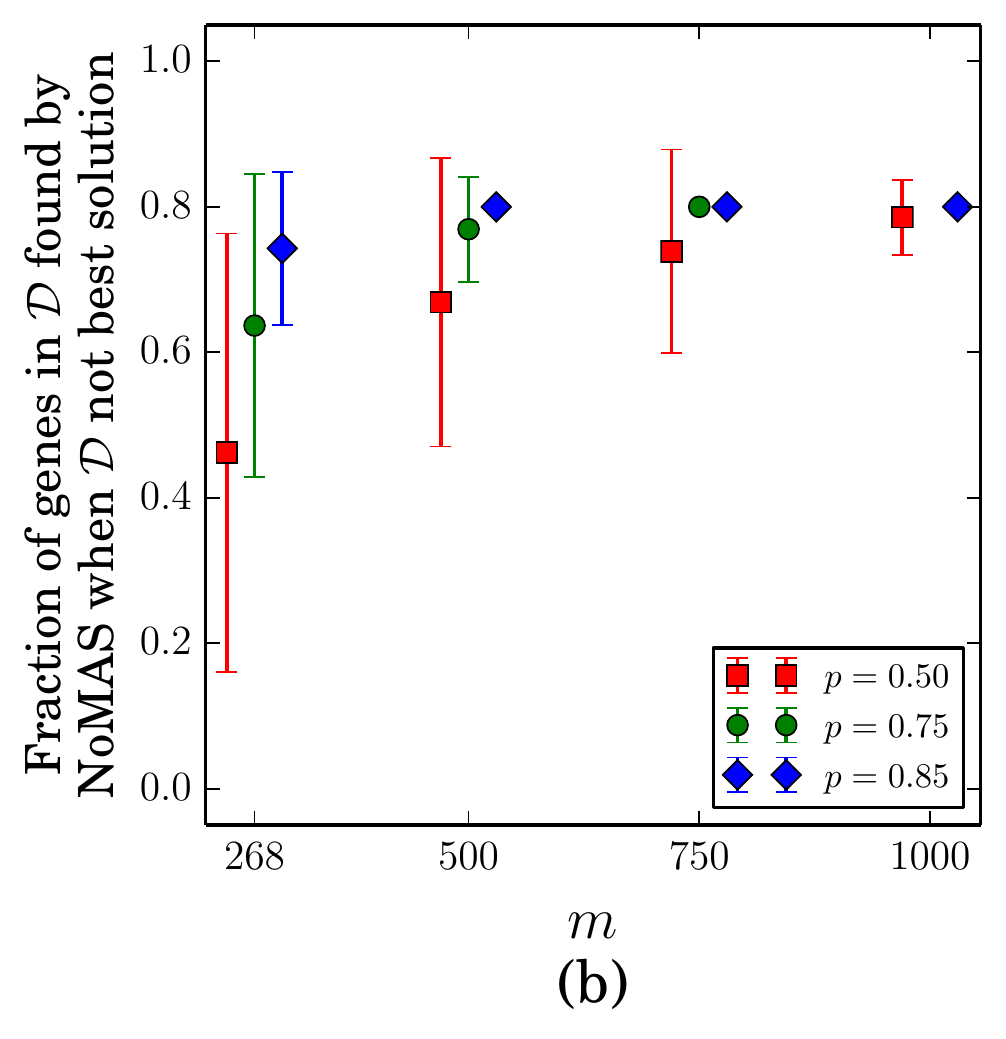}
\includegraphics[scale=0.40]{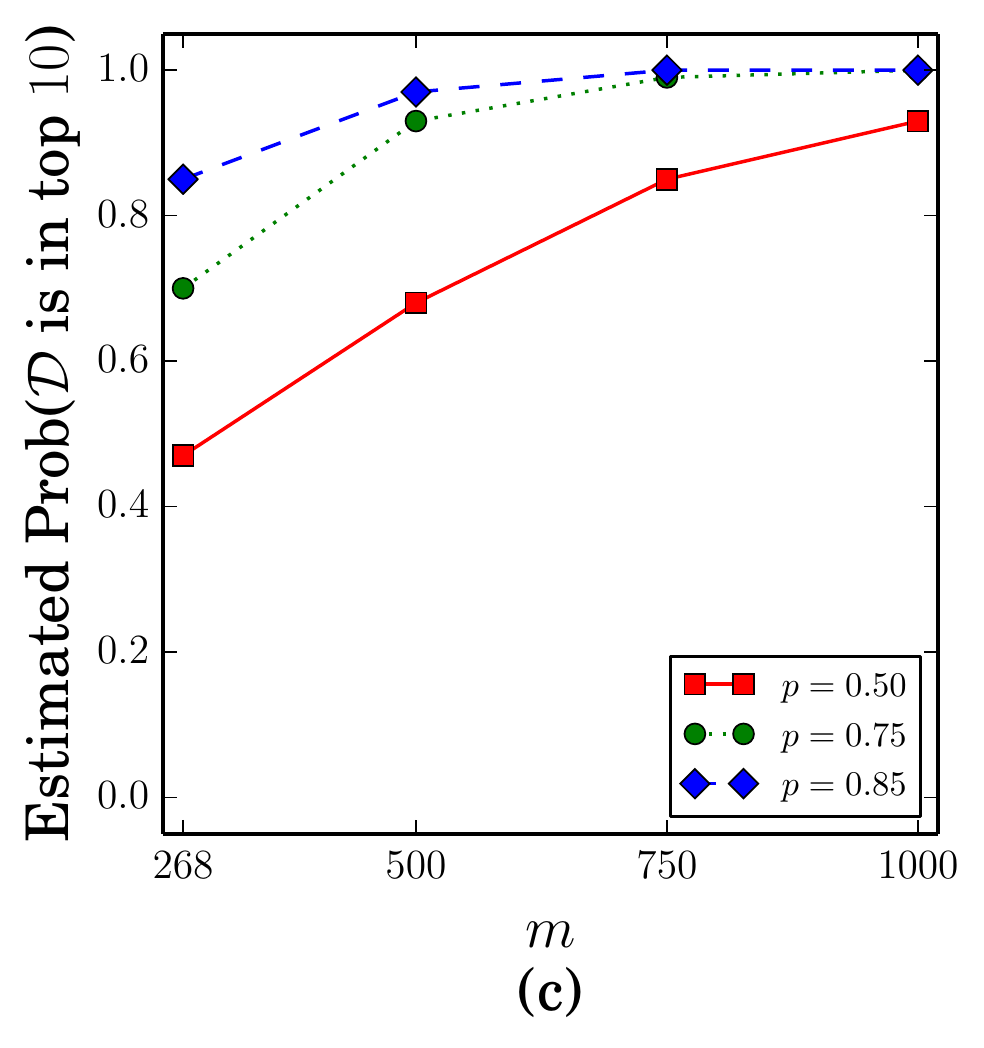}

\caption{Results of NoMAS on simulated data from the Planted Subnetwork Model. $100$ datasets were generated for each pair ($m$,$p$), where $m$  is the number of samples and for different probabilities $p$ of mutations in the set $\Pathway$ of genes associated with survival. (a) Probability that $\mathcal{D}$ is reported as the highest scoring solution by \algname. (b) Ratio of genes from the set $\mathcal{D}$ that are in the best solution when $\Pathway$ is not the highest scoring  solution by \algname. (c) Probability that $\mathcal{D}$ is among the top-10 solutions reported by \algname. All probabilities are estimated from the simulated datasets. }
\label{fig:sim1}
\end{figure}

We fixed $k=5$ and considered the values of $p \in \{0.5, 0.75, 0.85\}$ and $m \in\{268, 500, 750, 1000\}$. We kept the same ratio of censored observations as in GBM and chose the censored samples uniformly among all samples. For every pair  $(p,m)$ we performed $100$ simulations, running \algname\ on the dataset with $L=256$ color-coding iterations, and recorded whether \algname\ reported $\Pathway$ as the highest scoring subnetwork. Results are shown in Fig.~\ref{fig:sim1} (a).
For sample sizes similar to the currently available ones, \algname\ frequently reports $\Pathway$ as the highest scoring solutions when there is a quite strong association of $\Pathway$ with  survival ($p\ge 0.85$), but for $m=1000$ the highest scoring subnetwork reported by \algname\ is $\Pathway$ in $> 80\%$ of the cases even for $p=0.5$. Fig.~\ref{fig:sim1}(b) shows that even when \algname\ does not report $\Pathway$ as the highest scoring solution, the solution reported by \algname\ contains mostly genes that are in $\Pathway$, even for current sample size (e.g., on average $74\%$ of the genes in the $\Pathway$ are reported  by \algname\ for $m=268$ and $p=0.85$ even when $\Pathway$ is not the highest scoring solution by \algname). 
Finally, we assessed whether $\Pathway$ would be among the highest scoring solutions in the table $W$ computed by \algname: Fig~\ref{fig:sim1}(c) shows that by considering the top-10 solutions $W$ the chances to identify $\Pathway$ increase substantially even for $m=268$ and $p=0.5$, with most configurations having $>0.8$ probability of finding $\Pathway$ in the top-10 solutions by \algname. For a fixed $p=0.75$ and for each value of $m$ we assessed whether \algname\ identified the optimal solution even when it was not $\Pathway$ (an event not excluded in the Planted subnetwork Model) and found that for $m \geq 500$ \algname\ reported the optimal solution in $10$ out of $10$ cases (for $m=268$ \algname\ identified the optimal solution $9$ out of $10$ times).
These results show that \algname\ does indeed find the optimal solution in almost all cases even for sample sizes currently available (while the theoretical analysis of Section~\ref{sec:analysis} suggests that much larger sample sizes are required) and it can be used to identify $\Pathway$ or the majority of it by considering the top-10 highest scoring solutions.

The performance of \algname\ is affected when altering the ratio of samples that, with probability $p$, are mutated in $\mathcal{D}$. A higher ratio results in an increased performance of NoMAS, i.e. more cases arise in which $\mathcal{D}$ is the best subnetwork identified, and equally is $\mathcal{D}$ more frequently in the top $10$. When increasing the ratio (for example $\frac{m}{3}$ in stead of $\frac{m}{4}$) more samples are being mutated, and each gene in $\mathcal{D}$ receives more mutations (the gene responsible for each of the mutated samples is chosen uniformly at random). The increment in the association to survival is thus increased each time one of the genes is added to $\mathcal{D}$.

\subsection{Cancer data}
\label{sec:res_cancer}

We assessed the performance of \algname\ on the  GBM, OV, and LUSC datasets. We first assessed whether \algname\ identified the optimal solution by comparing the highest scoring solution reported by \algname\ with the one identified by using the exhaustive algorithm for $k=2,3,4,5$. In all cases we found that \algname\ does identify the optimal solution, while requiring much less running time compared to the exhaustive algorithm (Supplementary Fig.~\ref{fig:runtime}). For $k>5$ we could not run the exhaustive algorithm, while the runtime of \algname\ is still reasonable. The runtime of \algname\ can be greatly improved by using the parallelization strategy described in Section~\ref{sec:alg} (Supplementary Fig.~\ref{fig:parallel}). We therefore used \algname\ to find subnetworks of size $k=6$ and $k=8$.
We also considered two modifications of \algname\ that solve some easy cases where \algname\ may not identify the highest scoring solution due to its subnetwork merging strategy (see Appendix for a description and pseudo code of the modifications).
We run both modifications on GBM, OV, and LUSC for $k=6,8$ (using the same colorings used by the original version of \algname): in all cases the modified versions of \algname\ did not report subnetworks with higher scores than the ones from the original version of \algname.  We also note that the original version of \algname\ is significantly faster in practice than its two modifications (Supplementary Fig~\ref{fig:parallel}), and therefore we used the original version of \algname\ in the remaining experiments.

\newcommand{\naive}{\texttt{Greedy1}}
\newcommand{\fss}{\texttt{GreedyK}}
\newcommand{\dfs}{\texttt{GreedyDFS}}

We also compared \algname\ with three different greedy strategies for the \problemgraphname\ problem. All three algorithms build solutions starting from each node $u \in G$ in iterations by adding nodes to the current solution $\mathcal{S}$, and differ in the way they enlarge the current subnetwork $\mathcal{S}$ of size $1\leq i < k$. (See Appendix for a description of the greedy algorithms).
We run the three greedy algorithms on GBM, OV, and LUSC for $k=4,5,6,8$.
For each dataset we compared the resulting subnetworks with the ones identified by \algname. Results are shown in Fig.~\ref{fig:comp}.
In almost all cases we found that \algname\ discovered subnetworks with higher score than the subnetworks found by using greedy strategies, even if in some cases there is a greedy strategy that identifies the same subnetworks for all values of $k$. 
The difference in score increases as $k$ increases, showing the ability of \algname\ to discover better solutions for larger values of $k$ (see Supplementary Fig. \ref{fig:nomasGreedyTimes} for a running time comparison between NoMAS and the greedy algorithms). We also assessed whether the fact that greedy strategies discover lower scoring solutions than \algname\ has an impact on the estimate of the $p$-value in the permutational test. We considered the top-10 scoring solutions (corresponding  to 10 different starting nodes $u \in G$) discovered by the best greedy stategy in the GBM dataset, and compute the permutational $p$-value for each solution by generating 100 permuted datasets and either use the (same) greedy strategy for permuted data or use \algname\ for permuted data (using only $32$ iterations on the permuted data) 
Supplementary Fig.~\ref{fig:GBMpval} shows a comparison of the distribution of the $p$-values. As we can see, the greedy strategy incorrectly underestimate the permutational $p$-values for the solutions, due to the greedy algorithm not being able to identify solutions of score as high as \algname\ in the permuted datasets. The use of the greedy algorithms would then lead to both \begin{inparaenum} \item identify solutions in real data with lower association to survival compared to \algname\ and \item wrongly estimate their permutational $p$-value as more significant than it is.\end{inparaenum}

Finally, we compared \algname\ with the use of an (additive) score that sums single gene scores (similar to the ones used in~\cite{Vandin:2012aa}.  For each gene $g \in G$ we computed the $p$-value $p(g)$ for the association of $g$ with survival using the log-rank test, and define $a(\Set)= \sum_{g \in \Set} - \log_{10} p(g)$. We then partitioned the genes according to their association with increased survival or with decreased survival , and modified our algorithm to look for high scoring solutions in a partition using score $a(\Set)$. Results are in Fig~\ref{fig:comp}. We found that \algname\ outperforms the use of a single gene score, with a very large difference for certain values of the parameters.

\begin{figure}[ht]
\centering
\includegraphics[scale=0.44]{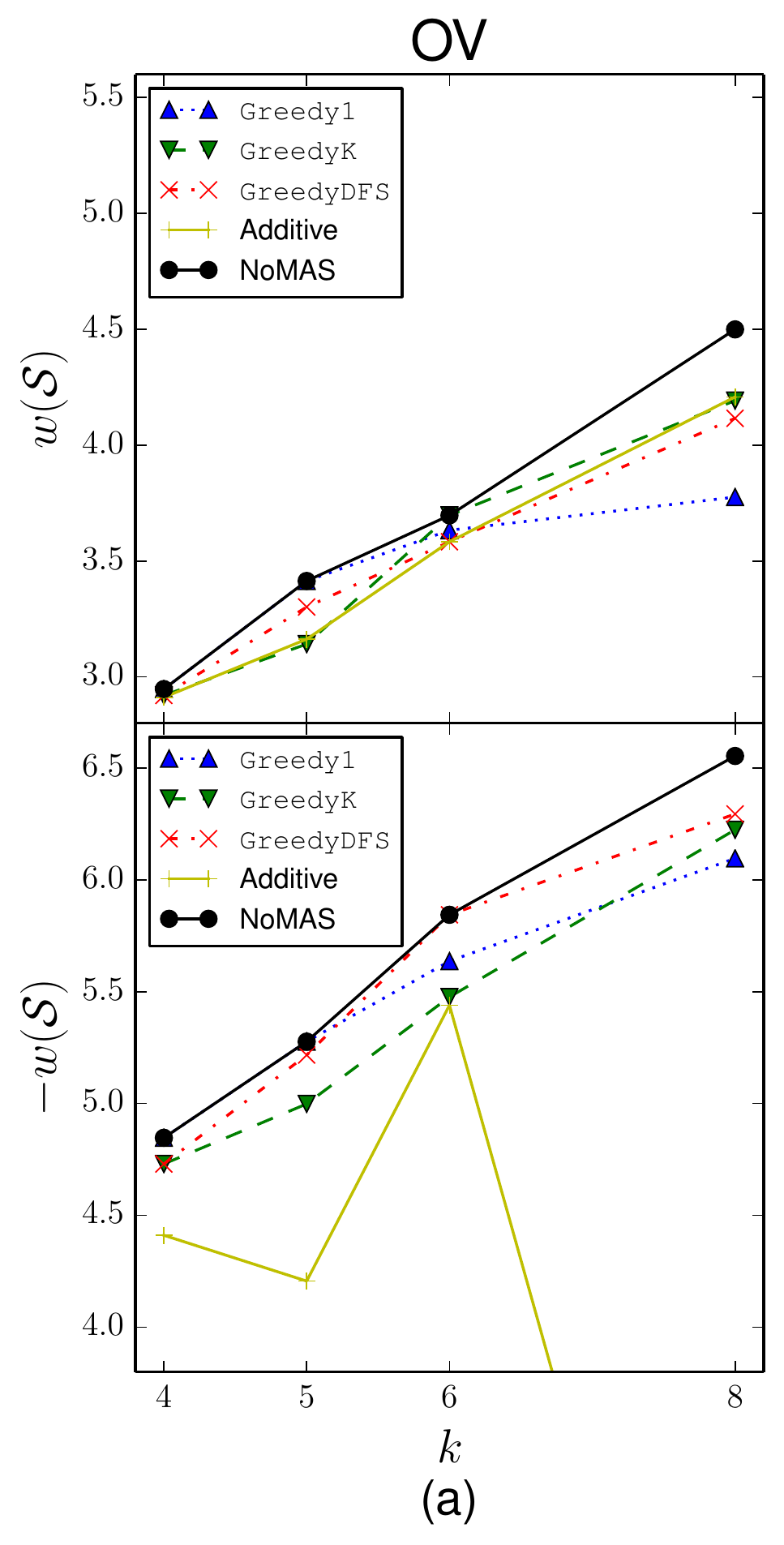}
\includegraphics[scale=0.44]{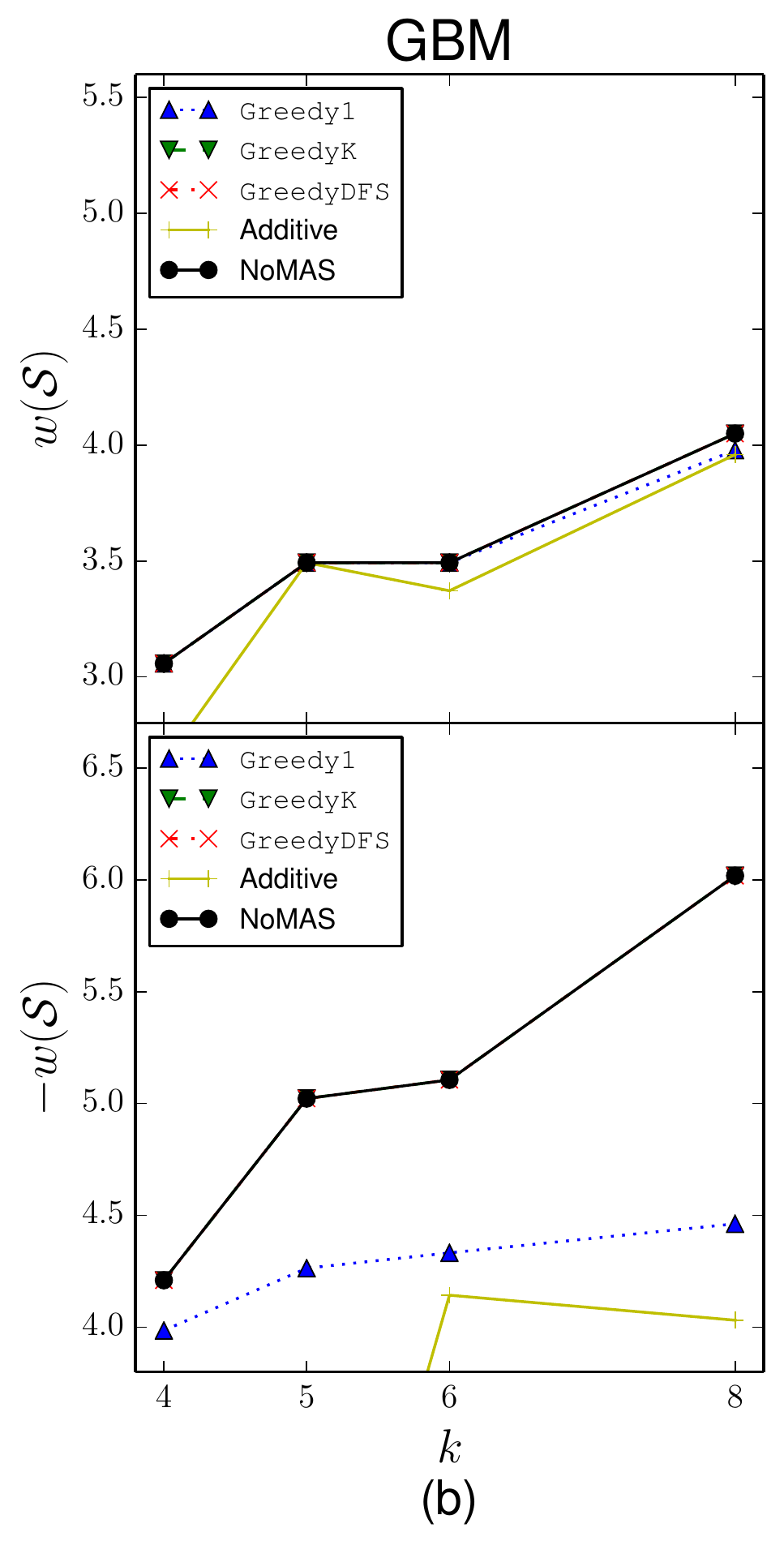}
\includegraphics[scale=0.44]{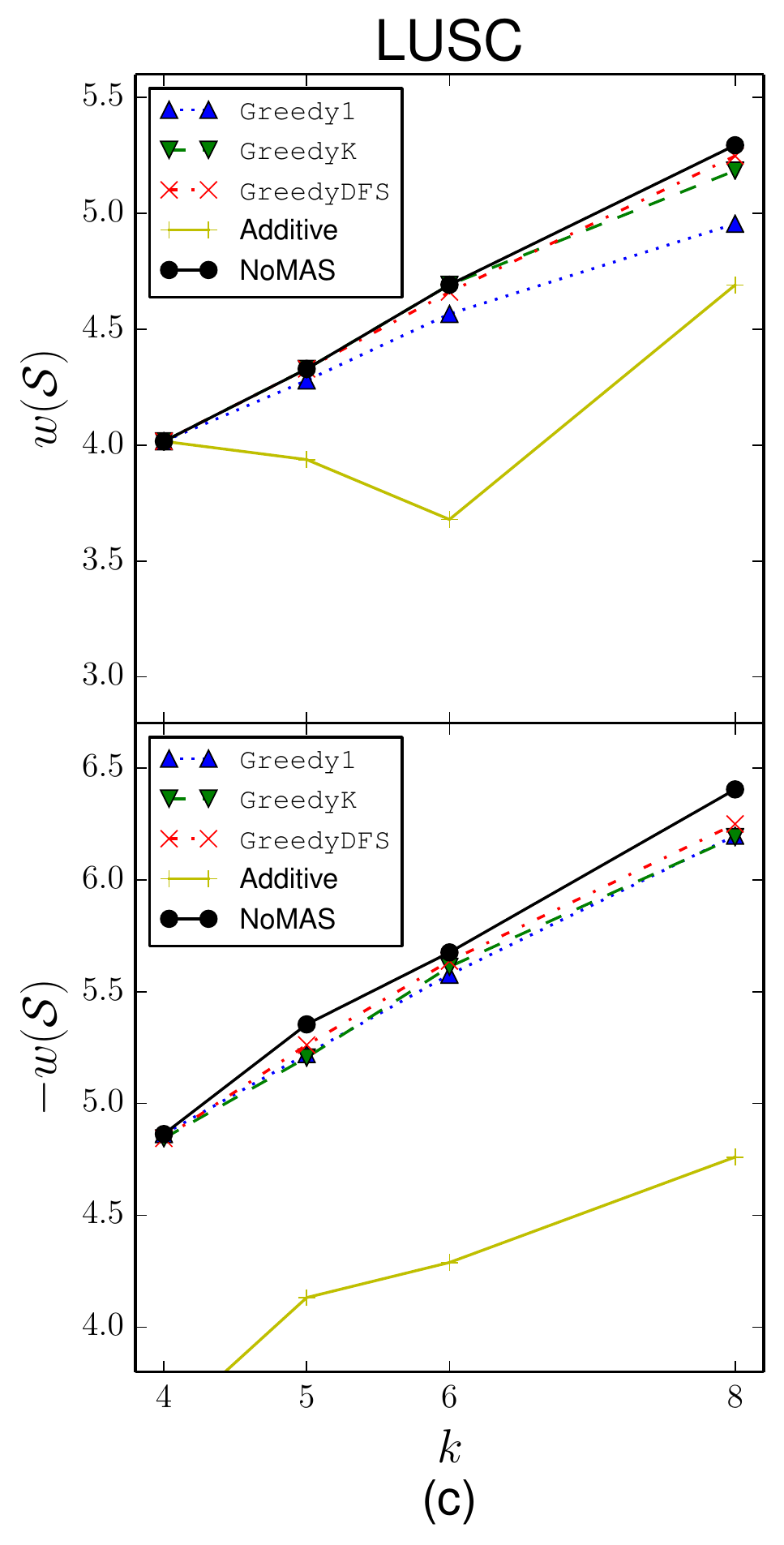}

\caption{Comparison of the normalized log-rank statistic of the best solution reported by \algname, by greedy algorithms (see Appendix for the description), and by the algorithm that uses an additive scoring function $a(\Set)$ (denote by ``additive" in the plots). To maintain readability we omit values above $-4.0$ when considering mutations associated with increased survival. For each datasets the results for the maximization of $w(\Set)$ (top panel) and the maximization of $-w(\Set)$ (bottom panel) are shown separately. (a) Results for GBM dataset. (b) Results for OV dataset. (c) Results for LUSC dataset.}
\label{fig:comp}
\end{figure}

We then considered the top-10 highest scoring subnetworks obtained from \algname\ on GBM and OV for $k=8$. For each subnetwork $\mathcal{S}$ we estimated the log-rank $p$-value and the permutational $p$-value as described in Section~\ref{sec:alg}. These subnetworks do not contain any \emph{gene} that would be reported as significant by single gene tests (corrected $p$=1), but they all show a high association with survival: in GBM, all subnetwork have log-rank $p$-value $<2 \times 10^{-6}$; in OV all subnetwork have $p$-value $<5\times 10^{-7}$. 

\begin{figure}[ht]
\centering
\includegraphics[scale=0.43]{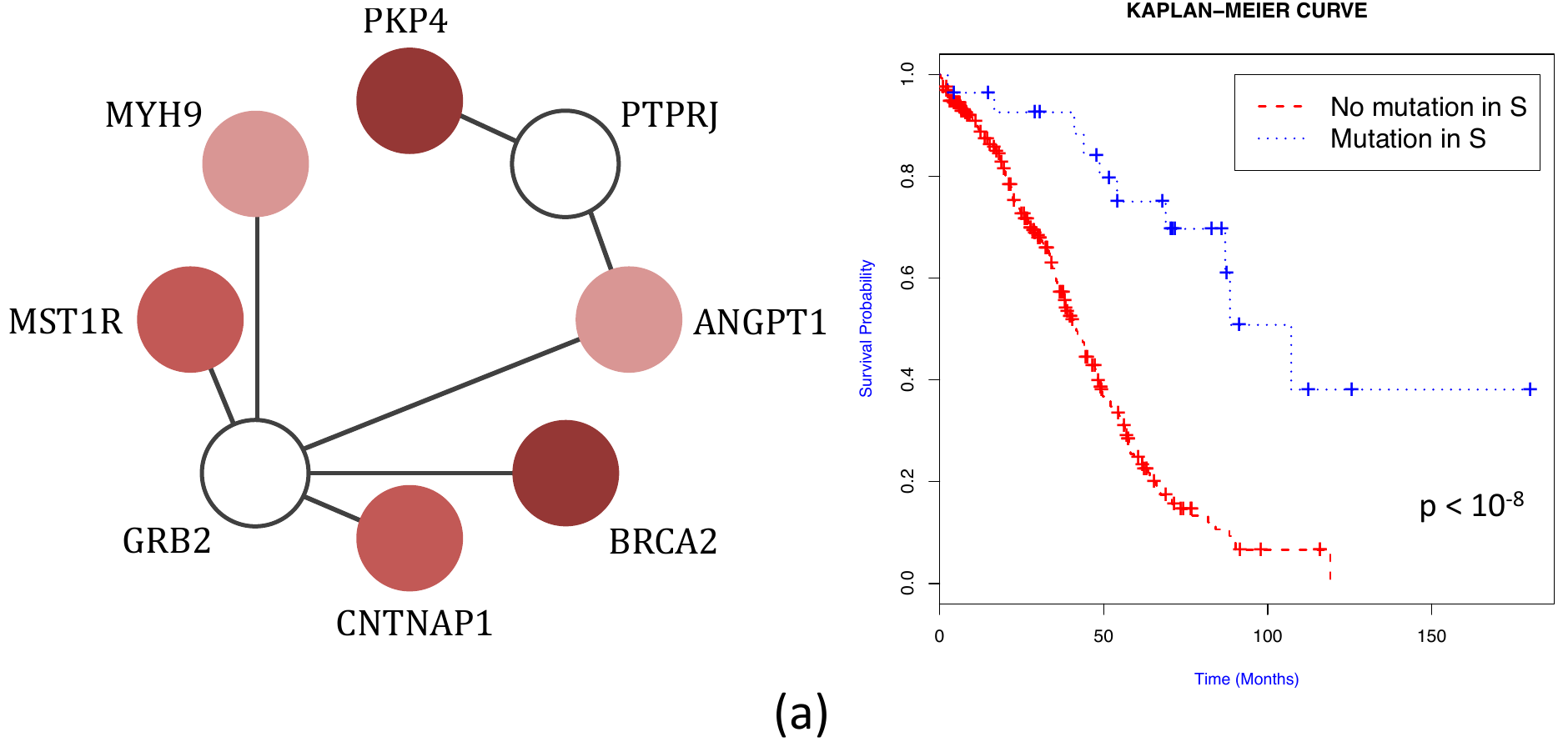}
\includegraphics[scale=0.43]{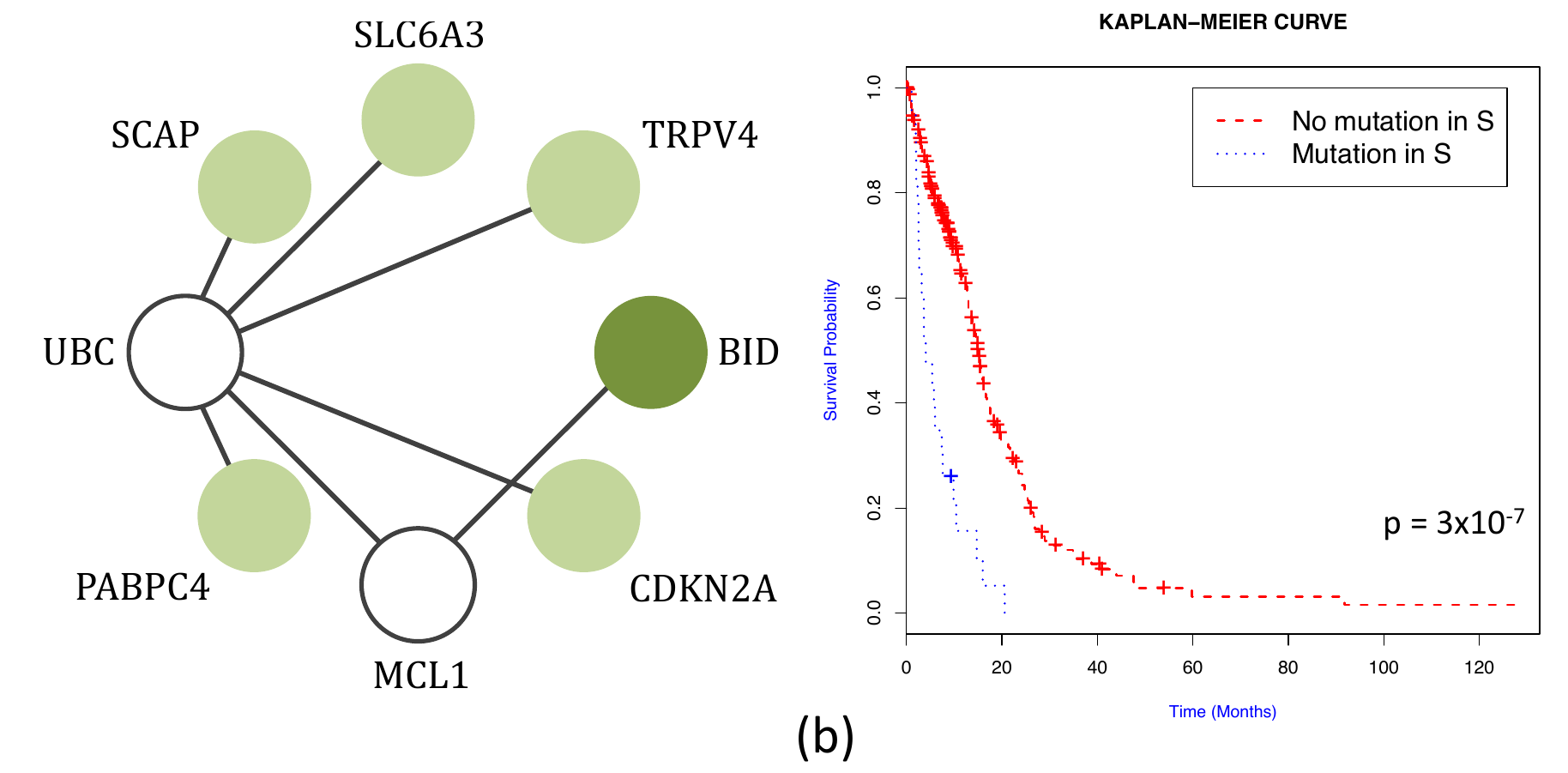}

\caption{Two subnetworks identified by \algname\ on real cancer data. (a) Subnetwork $S$ associated with survival in GBM, and Kaplan-Meier plot for the samples with mutations in $S$ vs samples with no mutation in $S$. Nodes in network are color coded according to their contribution to the normalized log-rank statistic of the network (dark color = higher contribution). White nodes have no mutations. The $p$-value from the log-rank test is shown. (b) Subnetwork $S$ associated with survival in OV, and Kaplan-Meier plot for the samples with mutations in $S$ vs samples with no mutation in $S$. Color-coding as in (a).}
\label{fig:OVandGBM}
\end{figure}
For example, in OV we identify a subnetwork (Fig~\ref{fig:OVandGBM}a) of 8 genes associated with increased survival including BRCA2, a known cancer gene  previously reported to have mutations associated with improved survival~\cite{Yang:2011aa} in ovarian cancer, MSTR1, previously reported as a novel prognostic marker in gastroesophageal adenocarcinoma~\cite{Catenacci:2011aa}, and MYH9, associated to metastasis and tumor invasion in gastric cancer~\cite{Liang:2011aa}. The permutational $p$-value for this subnetwork is $0.07$, and therefore it is unlikely this association is due to random variation. In GBM, we identify a subnetwork (Fig~\ref{fig:OVandGBM}b; permutational $p$-value $=0.2$) of 8 genes associated to decreased survival, including CDKN2A, an important component of the Rb pathway in GBM~\cite{Brennan:2013aa}, BID, a known cell death regulator, SCAP, involved in the molecular mechanisms of lipid metabolism in gliomas~\cite{Guo:2013aa}, TRPV4, previously identified as part of the cell migration mechanism~\cite{Fiorio-Pla:2013aa}, CARD6, involved in the antiapoptotic pathway and previously found to be associated to drug sensitivity in human glioblastoma multiforme cell lines~\cite{Halatsch:2009aa}, and IRS1, part of the known PI3K/PTEN cancer pathway~\cite{Parsons:2008aa}. These  results show that \algname\ identifies subnetworks associated with survival data in GBM and OV including known cancer genes and genes previously reported to be associated with survival, as well as genes that represent novel candidates for the association with survival.

\section{Conclusion}

In this work we study the problem of identifying subnetworks of a large gene-gene interaction network that are associated with survival using mutations from large cancer genomic studies. We formally define the associated
computational problem, that we call the \problemgraphname\ problem, by using as score for a subnetwork the test statistic of the log-rank test, one of the most widely used statistical tests to assess the significance in the difference in
survival among two populations. We prove 
that the \problemgraphname\ problem is NP-hard in general, and is NP-hard even when restricted to graphs with at least one node of large degree. We develop a new algorithm, \algname, based on the color-coding
technique, to efficiently identify high-scoring subnetworks associated with survival. We prove that even if our algorithm is not guaranteed to identify the optimal solution with the probability given by the color-coding technique (due the
non additivity of our scoring function), it does identify the optimal solution with the same guarantees given by the color-coding technique when the data comes from a reasonable model for mutations and independently of the survival
data. Using simulated data, we show that \algname\ is more efficient than the exhaustive algorithm while still identifying the optimal solution, and that our algorithm will identify subnetworks associated with survival when sample sizes larger than most currently available ones, but still reasonable, are available.

We use cancer data from three cancer studies from TCGA to compare \algname\ to approaches based on single gene scores and to greedy methods similar to ones proposed in the literature for the identification of 
subnetworks associated with survival and for other problems
on graphs. Our results show that \algname\ identify subnetworks with stronger association to survival compared to other approaches, and allows the correct estimation of $p$-values using a permutation test. Moreover,
in two datasets \algname\ identifies two subnetworks associated with survival containing genes previously reported to be important for prognosis in the same cancer type as well as novel genes, while no gene is significantly 
associated with survival when considered in isolation.

There are many directions in which this work can be extended. First, we only considered single nucleotide variants and indels in our analysis; we plan to extend our method to consider more complex variants (e.g., copy number aberrations and differential methylation) in the analysis.
Second, we believe that our algorithm and its analysis could be extended to the identification of subnetworks associated with clinical parameters other than survival time and to case-control studies, but substantial modifications to the algorithm
and to its analysis will be required. Third, this work considers the log-rank statistic as a measure of association with survival; another popular test in survival analysis is the use of Cox's regression model~\cite{Kalbfleisch:2002fk}. The two
tests are identical in the case of two populations, therefore our algorithm identifies subnetworks with high score w.r.t. Cox's regression model as well. However, Cox's regression model allows for the correction for covariates (e.g., gender,
age, etc.) in the analysis of survival data. A similar approach could be obtained by stratifying the patients in the log-rank test, but how to efficiently identify subnetworks, and in general combinations of genomic features,
 associated with survival while correcting for covariates remains a challenging open problem.
 
Finally, in this work we have restricted \algname\ to look for subnetworks of size at most $8$, due to the fact that \algname\ is exponential in $k$, and that the running time for identifying larger subnetworks with low error probabilities therefore is expensive (albeit significantly faster than an exhaustive enumeration).
However, we see that subnetworks of size $k$ reported by \algname\ are very likely to contain genes that overlap with reported subnetworks of size $< k$. Thus, solving a smaller problem, which is cheaper to compute, might provide us with parts of the solution to a bigger problem, i.e. indicate vertices of the gene-gene interaction network that are of interest when searching for larger subnetworks.
We can define a local search space $V^\prime \subset V$ around such interesting \textit{seed} vertices, for example as consisting of all the vertices reachable by at most $t$ edges from any seed vertex, for some parameter $t$ specifying the size/diversity of the search space.
By adjusting $t$ (as well as the number of seed vertices), this local search approach within a reduced vertex set will allow \algname\ to more efficiently find subnetworks of the current maximum size, but also to look for even larger subnetworks. The greedy algorithms are very fast and might therefore be good candidates for finding seed vertices. However, other methods for finding interesting areas of the gene-gene interaction network, such as using (or even combining) high scoring subnetworks identified by \algname\ for smaller problem sizes, can also be explored.

\section{Acknowledgements}
This work is supported, in part,  by
MIUR of Italy under project AMANDA
and by  NSF grant IIS-1247581.
The results presented in this manuscript are in whole or part based
upon data generated by the TCGA Research Network: \texttt{http://cancergenome.nih.gov/}.
This paper was selected for oral presentation at RECOMB 2016 and an abstract is published in the conference proceedings

\newpage
\bibliographystyle{plain}
\bibliography{survival_net}  

\begin{thebibliography}{10}

\bibitem{Alon:2008aa}
Noga Alon, Phuong Dao, Iman Hajirasouliha, Fereydoun Hormozdiari, and S~Cenk
  Sahinalp.
\newblock Biomolecular network motif counting and discovery by color coding.
\newblock {\em Bioinformatics}, 24(13):i241--9, Jul 2008.

\bibitem{alon1994color}
Noga Alon, Raphy Yuster, and Uri Zwick.
\newblock Color-coding: a new method for finding simple paths, cycles and other
  small subgraphs within large graphs.
\newblock In {\em Proceedings of the twenty-sixth annual ACM symposium on
  Theory of computing}, pages 326--335. ACM, 1994.

\bibitem{Brennan:2013aa}
Cameron~W Brennan, Roel G~W Verhaak, Aaron McKenna, Benito Campos, Houtan
  Noushmehr, Sofie~R Salama, Siyuan Zheng, Debyani Chakravarty, J~Zachary
  Sanborn, Samuel~H Berman, Rameen Beroukhim, Brady Bernard, Chang-Jiun Wu,
  Giannicola Genovese, Ilya Shmulevich, Jill Barnholtz-Sloan, Lihua Zou,
  Rahulsimham Vegesna, Sachet~A Shukla, Giovanni Ciriello, W~K Yung, Wei Zhang,
  Carrie Sougnez, Tom Mikkelsen, Kenneth Aldape, Darell~D Bigner, Erwin~G
  Van~Meir, Michael Prados, Andrew Sloan, Keith~L Black, Jennifer Eschbacher,
  Gaetano Finocchiaro, William Friedman, David~W Andrews, Abhijit Guha, Mary
  Iacocca, Brian~P O'Neill, Greg Foltz, Jerome Myers, Daniel~J Weisenberger,
  Robert Penny, Raju Kucherlapati, Charles~M Perou, D~Neil Hayes, Richard
  Gibbs, Marco Marra, Gordon~B Mills, Eric Lander, Paul Spellman, Richard
  Wilson, Chris Sander, John Weinstein, Matthew Meyerson, Stacey Gabriel,
  Peter~W Laird, David Haussler, Gad Getz, Lynda Chin, and {TCGA Research
  Network}.
\newblock The somatic genomic landscape of glioblastoma.
\newblock {\em Cell}, 155(2):462--77, Oct 2013.

\bibitem{Bruckner:2010aa}
Sharon Bruckner, Falk H{\"u}ffner, Richard~M Karp, Ron Shamir, and Roded
  Sharan.
\newblock Topology-free querying of protein interaction networks.
\newblock {\em J Comput Biol}, 17(3):237--52, Mar 2010.

\bibitem{Cancer-Genome-Atlas-Network:2015aa}
{Cancer Genome Atlas Network}.
\newblock Comprehensive genomic characterization of head and neck squamous cell
  carcinomas.
\newblock {\em Nature}, 517(7536):576--82, Jan 2015.

\bibitem{Cancer-Genome-Atlas-Research-Network:2011aa}
{Cancer Genome Atlas Research Network}.
\newblock Integrated genomic analyses of ovarian carcinoma.
\newblock {\em Nature}, 474(7353):609--15, Jun 2011.

\bibitem{Cancer-Genome-Atlas-Research-Network:2013aa}
{Cancer Genome Atlas Research Network}.
\newblock Comprehensive molecular characterization of clear cell renal cell
  carcinoma.
\newblock {\em Nature}, 499(7456):43--9, Jul 2013.

\bibitem{Cancer-Genome-Atlas-Research-Network:2014aa}
{Cancer Genome Atlas Research Network}.
\newblock Integrated genomic characterization of papillary thyroid carcinoma.
\newblock {\em Cell}, 159(3):676--90, Oct 2014.

\bibitem{Catenacci:2011aa}
Daniel V~T Catenacci, Gustavo Cervantes, Soheil Yala, Erik~A Nelson, Essam
  El-Hashani, Rajani Kanteti, Mohamed El~Dinali, Rifat Hasina, Johannes
  Br{\"a}gelmann, Tanguy Seiwert, Michele Sanicola, Les Henderson, Tatyana~A
  Grushko, Olufunmilayo Olopade, Theodore Karrison, Yung-Jue Bang, Woo~Ho Kim,
  Maria Tretiakova, Everett Vokes, David~A Frank, Hedy~L Kindler, Heather Huet,
  and Ravi Salgia.
\newblock Ron (mst1r) is a novel prognostic marker and therapeutic target for
  gastroesophageal adenocarcinoma.
\newblock {\em Cancer Biol Ther}, 12(1):9--46, Jul 2011.

\bibitem{chowdhury2011subnetwork}
Salim~A Chowdhury, Rod~K Nibbe, Mark~R Chance, and Mehmet Koyut{\"u}rk.
\newblock Subnetwork state functions define dysregulated subnetworks in cancer.
\newblock {\em Journal of Computational Biology}, 18(3):263--281, 2011.

\bibitem{Chuang:2007aa}
Han-Yu Chuang, Eunjung Lee, Yu-Tsueng Liu, Doheon Lee, and Trey Ideker.
\newblock Network-based classification of breast cancer metastasis.
\newblock {\em Mol Syst Biol}, 3:140, 2007.

\bibitem{dao2011optimally}
Phuong Dao, Kendric Wang, Colin Collins, Martin Ester, Anna Lapuk, and S~Cenk
  Sahinalp.
\newblock Optimally discriminative subnetwork markers predict response to
  chemotherapy.
\newblock {\em Bioinformatics}, 27(13):i205--i213, 2011.

\bibitem{Das:2012aa}
Jishnu Das and Haiyuan Yu.
\newblock Hint: High-quality protein interactomes and their applications in
  understanding human disease.
\newblock {\em BMC Syst Biol}, 6:92, 2012.

\bibitem{Dees:2012aa}
Nathan~D Dees, Qunyuan Zhang, Cyriac Kandoth, Michael~C Wendl, William
  Schierding, Daniel~C Koboldt, Thomas~B Mooney, Matthew~B Callaway, David
  Dooling, Elaine~R Mardis, Richard~K Wilson, and Li~Ding.
\newblock Music: identifying mutational significance in cancer genomes.
\newblock {\em Genome Res}, 22(8):1589--98, Aug 2012.

\bibitem{Fiorio-Pla:2013aa}
Alessandra Fiorio~Pla and Dimitra Gkika.
\newblock Emerging role of trp channels in cell migration: from tumor
  vascularization to metastasis.
\newblock {\em Front Physiol}, 4:311, 2013.

\bibitem{Garraway:2013uq}
Levi~A Garraway and Eric~S Lander.
\newblock Lessons from the cancer genome.
\newblock {\em Cell}, 153(1):17--37, Mar 2013.

\bibitem{Guo:2013aa}
Deliang Guo, Erica~Hlavin Bell, and Arnab Chakravarti.
\newblock Lipid metabolism emerges as a promising target for malignant glioma
  therapy.
\newblock {\em CNS Oncol}, 2(3):289--99, May 2013.

\bibitem{Halatsch:2009aa}
Marc-Eric Halatsch, Sarah L{\"o}w, Kay Mursch, Thomas Hielscher, Ursula
  Schmidt, Andreas Unterberg, Vassilios~I Vougioukas, and Friedrich Feuerhake.
\newblock Candidate genes for sensitivity and resistance of human glioblastoma
  multiforme cell lines to erlotinib. laboratory investigation.
\newblock {\em J Neurosurg}, 111(2):211--8, Aug 2009.

\bibitem{Hofree:2013aa}
Matan Hofree, John~P Shen, Hannah Carter, Andrew Gross, and Trey Ideker.
\newblock Network-based stratification of tumor mutations.
\newblock {\em Nat Methods}, 10(11):1108--15, Nov 2013.

\bibitem{Hormozdiari:2015aa}
Fereydoun Hormozdiari, Osnat Penn, Elhanan Borenstein, and Evan~E Eichler.
\newblock The discovery of integrated gene networks for autism and related
  disorders.
\newblock {\em Genome Res}, 25(1):142--54, Jan 2015.

\bibitem{Kalbfleisch:2002fk}
John Kalbfleisch and Ross Prentice.
\newblock {\em The Statistical Analysis of Failure Time Data (Wiley Series in
  Probability and Statistics)}.
\newblock Wiley-Interscience, 2 edition, 2002.

\bibitem{Kandoth:2013aa}
Cyriac Kandoth, Michael~D McLellan, Fabio Vandin, Kai Ye, Beifang Niu, Charles
  Lu, Mingchao Xie, Qunyuan Zhang, Joshua~F McMichael, Matthew~A Wyczalkowski,
  Mark D~M Leiserson, Christopher~A Miller, John~S Welch, Matthew~J Walter,
  Michael~C Wendl, Timothy~J Ley, Richard~K Wilson, Benjamin~J Raphael, and
  Li~Ding.
\newblock Mutational landscape and significance across 12 major cancer types.
\newblock {\em Nature}, 502(7471):333--9, Oct 2013.

\bibitem{kim2015memcover}
Yoo-Ah Kim, Dong-Yeon Cho, Phuong Dao, and Teresa~M Przytycka.
\newblock Memcover: integrated analysis of mutual exclusivity and functional
  network reveals dysregulated pathways across multiple cancer types.
\newblock {\em Bioinformatics}, 31(12):i284--i292, 2015.

\bibitem{Lawrence:2013aa}
Michael~S Lawrence, Petar Stojanov, Paz Polak, Gregory~V Kryukov, Kristian
  Cibulskis, Andrey Sivachenko, Scott~L Carter, Chip Stewart, Craig~H Mermel,
  Steven~A Roberts, Adam Kiezun, Peter~S Hammerman, Aaron McKenna, Yotam Drier,
  Lihua Zou, Alex~H Ramos, Trevor~J Pugh, Nicolas Stransky, Elena Helman,
  Jaegil Kim, Carrie Sougnez, Lauren Ambrogio, Elizabeth Nickerson, Erica
  Shefler, Maria~L Cort{\'e}s, Daniel Auclair, Gordon Saksena, Douglas Voet,
  Michael Noble, Daniel DiCara, Pei Lin, Lee Lichtenstein, David~I Heiman,
  Timothy Fennell, Marcin Imielinski, Bryan Hernandez, Eran Hodis, Sylvan Baca,
  Austin~M Dulak, Jens Lohr, Dan-Avi Landau, Catherine~J Wu, Jorge
  Melendez-Zajgla, Alfredo Hidalgo-Miranda, Amnon Koren, Steven~A McCarroll,
  Jaume Mora, Ryan~S Lee, Brian Crompton, Robert Onofrio, Melissa Parkin, Wendy
  Winckler, Kristin Ardlie, Stacey~B Gabriel, Charles W~M Roberts, Jaclyn~A
  Biegel, Kimberly Stegmaier, Adam~J Bass, Levi~A Garraway, Matthew Meyerson,
  Todd~R Golub, Dmitry~A Gordenin, Shamil Sunyaev, Eric~S Lander, and Gad Getz.
\newblock Mutational heterogeneity in cancer and the search for new
  cancer-associated genes.
\newblock {\em Nature}, 499(7457):214--8, Jul 2013.

\bibitem{Leiserson:2015aa}
Mark D~M Leiserson, Fabio Vandin, Hsin-Ta Wu, Jason~R Dobson, Jonathan~V
  Eldridge, Jacob~L Thomas, Alexandra Papoutsaki, Younhun Kim, Beifang Niu,
  Michael McLellan, Michael~S Lawrence, Abel Gonzalez-Perez, David Tamborero,
  Yuwei Cheng, Gregory~A Ryslik, Nuria Lopez-Bigas, Gad Getz, Li~Ding, and
  Benjamin~J Raphael.
\newblock Pan-cancer network analysis identifies combinations of rare somatic
  mutations across pathways and protein complexes.
\newblock {\em Nat Genet}, 47(2):106--14, Feb 2015.

\bibitem{Leiserson:2015ab}
Mark D~M Leiserson, Hsin-Ta Wu, Fabio Vandin, and Benjamin~J Raphael.
\newblock Comet: a statistical approach to identify combinations of mutually
  exclusive alterations in cancer.
\newblock {\em Genome Biol}, 16:160, 2015.

\bibitem{Liang:2011aa}
Shuli Liang, Lijie He, Xiaodi Zhao, Yu~Miao, Yong Gu, Changcun Guo, Zengfu Xue,
  Weijia Dou, Fengrong Hu, Kaichun Wu, Yongzhan Nie, and Daiming Fan.
\newblock Microrna let-7f inhibits tumor invasion and metastasis by targeting
  myh9 in human gastric cancer.
\newblock {\em PLoS One}, 6(4):e18409, 2011.

\bibitem{pmid5910392}
N.~Mantel.
\newblock {{E}valuation of survival data and two new rank order statistics
  arising in its consideration}.
\newblock {\em Cancer Chemother Rep}, 50:163--170, 1966.

\bibitem{maxwell2014efficiently}
Sean Maxwell, Mark~R Chance, and Mehmet Koyut{\"u}rk.
\newblock Efficiently enumerating all connected induced subgraphs of a large
  molecular network.
\newblock In {\em Algorithms for Computational Biology}, pages 171--182.
  Springer, 2014.

\bibitem{Parsons:2008aa}
D~Williams Parsons, Si{\^a}n Jones, Xiaosong Zhang, Jimmy Cheng-Ho Lin,
  Rebecca~J Leary, Philipp Angenendt, Parminder Mankoo, Hannah Carter, I-Mei
  Siu, Gary~L Gallia, Alessandro Olivi, Roger McLendon, B~Ahmed Rasheed,
  Stephen Keir, Tatiana Nikolskaya, Yuri Nikolsky, Dana~A Busam, Hanna Tekleab,
  Luis~A Diaz, Jr, James Hartigan, Doug~R Smith, Robert~L Strausberg, Suely
  Kazue~Nagahashi Marie, Sueli Mieko~Oba Shinjo, Hai Yan, Gregory~J Riggins,
  Darell~D Bigner, Rachel Karchin, Nick Papadopoulos, Giovanni Parmigiani, Bert
  Vogelstein, Victor~E Velculescu, and Kenneth~W Kinzler.
\newblock An integrated genomic analysis of human glioblastoma multiforme.
\newblock {\em Science}, 321(5897):1807--12, Sep 2008.

\bibitem{patel2013network}
Vishal~N Patel, Giridharan Gokulrangan, Salim~A Chowdhury, Yanwen Chen,
  Andrew~E Sloan, Mehmet Koyut{\"u}rk, Jill Barnholtz-Sloan, and Mark~R Chance.
\newblock Network signatures of survival in glioblastoma multiforme.
\newblock {\em PLoS Comput Biol}, 9(9):e1003237, 2013.

\bibitem{citeulike:7445010}
R.~Peto and J.~Peto.
\newblock {Asymptotically efficient rank invariant test procedures}.
\newblock {\em J Roy Stat Soc Ser A}, 135:185+, 1972.

\bibitem{Raphael:2014aa}
Benjamin~J Raphael, Jason~R Dobson, Layla Oesper, and Fabio Vandin.
\newblock Identifying driver mutations in sequenced cancer genomes:
  computational approaches to enable precision medicine.
\newblock {\em Genome Med}, 6(1):5, 2014.

\bibitem{Reimand:2013aa}
J{\"u}ri Reimand and Gary~D Bader.
\newblock Systematic analysis of somatic mutations in phosphorylation signaling
  predicts novel cancer drivers.
\newblock {\em Mol Syst Biol}, 9:637, 2013.

\bibitem{shrestha2014hit}
Raunak Shrestha, Ermin Hodzic, Jake Yeung, Kendric Wang, Thomas Sauerwald,
  Phuong Dao, Shawn Anderson, Himisha Beltran, Mark~A Rubin, Colin~C Collins,
  et~al.
\newblock Hit'ndrive: Multi-driver gene prioritization based on hitting time.
\newblock In {\em Research in Computational Molecular Biology}, pages 293--306.
  Springer, 2014.

\bibitem{smyth1992information}
Padhraic Smyth and Rodney~M Goodman.
\newblock An information theoretic approach to rule induction from databases.
\newblock {\em Knowledge and Data Engineering, IEEE Transactions on},
  4(4):301--316, 1992.

\bibitem{Tamborero:2013aa}
David Tamborero, Abel Gonzalez-Perez, Christian Perez-Llamas, Jordi Deu-Pons,
  Cyriac Kandoth, J{\"u}ri Reimand, Michael~S Lawrence, Gad Getz, Gary~D Bader,
  Li~Ding, and Nuria Lopez-Bigas.
\newblock Comprehensive identification of mutational cancer driver genes across
  12 tumor types.
\newblock {\em Sci Rep}, 3:2650, 2013.

\bibitem{Vandin:2012aa}
Fabio Vandin, Patrick Clay, Eli Upfal, and Benjamin~J Raphael.
\newblock Discovery of mutated subnetworks associated with clinical data in
  cancer.
\newblock {\em Pac Symp Biocomput}, pages 55--66, 2012.

\bibitem{Vandin:2015aa}
Fabio Vandin, Alexandra Papoutsaki, Benjamin~J Raphael, and Eli Upfal.
\newblock Accurate computation of survival statistics in genome-wide studies.
\newblock {\em PLoS Comput Biol}, 11(5):e1004071, May 2015.

\bibitem{Vandin:2012ab}
Fabio Vandin, Eli Upfal, and Benjamin~J Raphael.
\newblock De novo discovery of mutated driver pathways in cancer.
\newblock {\em Genome Res}, 22(2):375--85, Feb 2012.

\bibitem{Vogelstein:2013fk}
Bert Vogelstein, Nickolas Papadopoulos, Victor~E Velculescu, Shibin Zhou,
  Luis~A Diaz, Jr, and Kenneth~W Kinzler.
\newblock Cancer genome landscapes.
\newblock {\em Science}, 339(6127):1546--58, Mar 2013.

\bibitem{Wu:2012aa}
Guanming Wu and Lincoln Stein.
\newblock A network module-based method for identifying cancer prognostic
  signatures.
\newblock {\em Genome Biol}, 13(12):R112, 2012.

\bibitem{Yang:2011aa}
Da~Yang, Sofia Khan, Yan Sun, Kenneth Hess, Ilya Shmulevich, Anil~K Sood, and
  Wei Zhang.
\newblock Association of brca1 and brca2 mutations with survival, chemotherapy
  sensitivity, and gene mutator phenotype in patients with ovarian cancer.
\newblock {\em JAMA}, 306(14):1557--65, Oct 2011.

\bibitem{Yu:2011aa}
Haiyuan Yu, Leah Tardivo, Stanley Tam, Evan Weiner, Fana Gebreab, Changyu Fan,
  Nenad Svrzikapa, Tomoko Hirozane-Kishikawa, Edward Rietman, Xinping Yang,
  Julie Sahalie, Kourosh Salehi-Ashtiani, Tong Hao, Michael~E Cusick, David~E
  Hill, Frederick~P Roth, Pascal Braun, and Marc Vidal.
\newblock Next-generation sequencing to generate interactome datasets.
\newblock {\em Nat Methods}, 8(6):478--80, Jun 2011.

\end{thebibliography}

\newpage
\appendix
{\huge Appendix}

\section{Proof Sketches}

Note that in Equation~\ref{eq:nlr}, the $j$-th term in the sum depends on the values of $x_1,\dots,x_{j-1}$. Our results use the fact that $V(\mathbf{x})$ can be computed as a sum of weights, one for each sample, that are independent of the value of the entries of $\mathbf{x}$.

\setcounter{proposition}{0}

\begin{proposition}
\label{prop:weightedlogrank}
For $i=1,2,\dots, m$, let $w_i = c_i  - \sum_{j=1}^{i}\frac{c_j}{m-j+1}$. Then
\begin{equation*}
\sum_{j=1}^m c_j \left(x_j -  \frac{m_1 -\sum_{i=1}^{j-1} x_i}{m-j+1}\right) = \sum_{j=1}^m w_j x_j.
\end{equation*}
\end{proposition}
\begin{proof}[Proof (Sketch)]
\begin{eqnarray*}
\sum_{j=1}^m w_j x_j & = & \sum_{j=1}^{m} \left( c_j - \sum_{i=1}^j \frac{c_i}{m-i+1} \right) x_j\\
& = & \left( \sum_{j=1}^m c_j x_j\right) - \left( \sum_{j=1}^m x_j \sum_{i=1}^j \frac{c_i}{m-i+1}\right)\\
& = & \left( \sum_{j=1}^m c_j x_j\right) - \left(\sum_{j:c_j \neq 0} \frac{\sum_{i=j}^m x_i}{m-j+1} \right)\\
& = & \left( \sum_{j=1}^m c_j x_j\right) - \left(\sum_{j=1}^m c_i \frac{m_1 - \sum_{i=1}^{j-1} x_i}{m-j+1} \right)\\
& = & \sum_{j=1}^m c_j \left(x_j -  \frac{m_1 -\sum_{i=1}^{j-1} x_i}{m-j+1}\right)  = V(\mathbf{x}).
\end{eqnarray*}
\end{proof}

\setcounter{theorem}{0}

\begin{theorem}
The \problemname\  problem is NP-hard.
\end{theorem}
\begin{proof}[Proof (Sketch)]
The reduction is from the minimum set cover problem. In particular, we will show that if we can find a set $\mathcal{S}$ with $|\mathcal{S}|$ maximizing $w'(\mathcal{S})$ in polynomial time, then we can test (in polynomial time) if there is a set cover of cardinality $k$. This implies that one could find the size of the minimum set cover in polynomial time, that is an NP-hard problem.

In the minimum set cover problem, one is given elements $e_1,\dots, e_n$, where each element $e_i, 1 \le i \le n$ is a subset of a universe set $\mathcal{U}$, with $|\mathcal{U}| = m$. The goal is to find the minimum cardinality subset $\mathcal{C} \subset \{ e_1, \dots, e_n\}$ such that $\cup_{e \in \mathcal{C}} = \mathcal{U}$.

Given an instance of the minimum set cover problem, we build an instance of the \problemname\ as follows. For each element $e_i, 1 \le i \le n$, we have a gene $g_i$, with $\mathcal{G} = \{ g_1, \dots, g_n\}$. The set $\mathcal{P}$ of patients has cardinality $4 |\mathcal{U}|$. $\mathcal{P}$ is partitioned into two sets $\mathcal{P}_1$ and $\mathcal{P}_2$, with $\mathcal{P}_1\cap \mathcal{P}_2 = \emptyset$ and $\mathcal{P} = \mathcal{P}_1 \cup \mathcal{P}_2$. Moreover we have $|\mathcal{P}_1| = \mathcal{U}$ and $|\mathcal{P}_2| = 3 \mathcal{U}$, and the survival time of all patients in $\mathcal{P}_1$ is lower than the survival time of all patients in $\mathcal{P}_2$. In addition, no patient of $\mathcal{P}_1$ is censored, while all patients in $\mathcal{P}_2$ are censored. The patients of $\mathcal{P}_1$ correspond to the elements of $\mathcal{U}$, and gene $g_i$ is mutated in patients $e_i \subset \mathcal{U}$. 

We now show that there is a minimum set cover of cardinality $k$ if and only if $\max_{\mathcal{S} \subset \mathcal{G}, |\mathcal{S}|=k} w'(\mathcal{S})= \frac{\frac{m}{4} - \sum_{j=1}^{m/4} \frac{1}{m-j+1}}{\sqrt{9 m^2/16}}$. In particular, we will show that the maximum log-rank statistic is obtained when $x_i=1$ for all $1 \le i \le \frac{m}{4}$ and $x_i=0$ for all $\frac{m}{4} < i \le n$, that can be achieved if and only if there is a set cover of cardinality $k$.  (Note that $m$ is divisible by $4$ by construction.)

To prove the above, it is enough to show the following:
\begin{enumerate}[i)]
\item $w_i> 0$ for $ 1 \le i \le \frac{m}{4}$;\label{proof:i}
\item for a fixed $m_1 \le \frac{m}{4}$, the maximum weight is given by $\frac{\sum_{i=1}^{m_1} w_i}{\sqrt{m_1(m-m_1)}}$;
\item for all $1\le j \le \frac{m}{4}-1$: $\frac{\sum_{i=1}^{j} w_i}{\sqrt{j(m-j)}} \le \frac{\sum_{i=1}^{j+1} w_i}{\sqrt{(j+1)(m-(j+1))}}$.
\end{enumerate} 

We first  note that $w_i > w_{i+1}$ for $1\le i < \frac{m}{4}$.

\begin{equation*}
 w_i -w_{i+1} =  1  - \sum_{j=1}^{i}\frac{1}{m-j+1} - 1  + \sum_{j=1}^{i+1}\frac{1}{m-j+1}  = \frac{1}{m-i} > 0.
 \end{equation*}
 
To prove \ref{proof:i}) above it is then enough to prove that $w_{\frac{m}{4}} > 0$.  

\begin{equation*}
w_{\frac{m}{4}} = 1  - \sum_{j=1}^{m/4} \frac{1}{m-j+1} = 1  - \sum_{j=\frac{3m}{4}+1}^{m} \frac{1}{j} = 1 - H(m)+H\left(\frac{3m}{4}\right).
\end{equation*}
where $H(m)$ is the $m$-th harmonic number. Since $H(m) \le \ln m + \gamma + \frac{1}{2m}$, with $\gamma \le 0.58$ constant, for $m$ large enough, and $H(m) \ge \ln m$, we have:
\begin{equation}
w_{\frac{m}{4}} \ge 1 - \ln m - \gamma  - \frac{1}{2m} + \ln{\frac{3m}{4}} \ge 1 - \ln \frac{4}{3} - \gamma - \frac{1}{2m}  > 0.1 - \frac{1}{m} > 0
\end{equation}

for $m$ large enough.

ii) follows immediately from i) and from $w_i > w_{i+1}$ for $1\le i < \frac{m}{4}$ (since fixed $m_1$, the denominator $\sqrt{m_1 (m-m_1)}$ is fixed).

iii) can be proved by induction.
\end{proof}

\begin{theorem}
The \problemgraphname\ problem on graphs with at least one node of degree $\BO{n^\frac{1}{c}}$, where $c>1$ is constant, is NP-hard.
\end{theorem}
\begin{proof}[Proof (Sketch)]
Take an instance of set cover with $n$ elements. We can ``encode'' it in the neighbours of a node of degree $n$ in a graph with $n^c$ vertices, where $c>0$ is a constant, using the same scheme used for Theorem~\ref{thm:nphard}. All other vertices have no mutations. Note that the reduction is polynomial.
\end{proof}

\begin{proposition}
For every $k \geq 3$ there is a family of instances of the \problemgraphname\ problem and colorings for which \opt\ is not found by our algorithm even if it is colorful.
\end{proposition}

\begin{proof}[Proof (Sketch)]
Let the number of samples be $n = 8(k-1)$. The censoring information $\mathbf{c}$ is such that $c_i = 1$ for $1 \leq i \leq \frac{n}{4}$ and $c_j = 0$ for $\frac{n}{4}+1\leq j \leq n$. From Theorem \ref{thm:nphard} we get that all weights $w_i > 0$ for $1 \leq i \leq \frac{n}{4}$.
Let $\mathcal{I}$ be a tree with one internal vertex $v_0$ and $k+1$ leaf vertices $\{v_1, v_2, \dots, v_{k-1}, \bar{v_1}, \bar{v_2}\}$. Consider a coloring $\mathcal{C}$ in which $\mathcal{C}(v_i)$ are distinct for $0 \leq i \leq k-1$ and $\mathcal{C}(v_j) = \mathcal{C}(\bar{v_j})$ for $1\leq j \leq 2$.
Let $\sigma(v)$ be the set of weights for vertex $v$, i.e containing a weight for each sample mutated in the gene associated with $v$. Assign the weights such that $\sigma(v_0) = \emptyset$, $\sigma(v_i) = \{w_i, w_{k-1+i}\}$ and $\sigma(\bar{v_i}) = \{w_1, w_2\}$.
Note that for any $k \geq 3$ the optimal connected subnetwork \opt\ $= \mathcal{S} = \{v_0, v_1, \dots, v_{k-1}\}$ since $\sigma(\mathcal{S}) = \{w_1, w_2, \dots, w_{n/4}\}$. By construction \opt\ is colorful.\\

The idea of the construction is to have two \textit{bad} colors. A color $c$ is bad if it is assigned to two vertices. The vertex in \opt\ with color $c$ is a \textit{good} vertex, while the vertex with color $c$ not in \opt\ is a \textit{bad} vertex. In our construction $v_1$ and $v_2$ are good vertices and $\bar{v_1}$ and $\bar{v_2}$ are bad vertices.
Recall that our algorithm combines two subnetworks that are connected by an edge, thus every subnetwork of size $\ell$ must be a combination of a leaf $v_i$ and some subnetwork $W(T, v_0)$ of size $\ell-1$. 
To generate \opt, at some point we will have that $v_i$ is one of the good vertices while $W(T, v_0)$ contains the other good vertex. We will show that this cannot happen. In particular we argue that $W(T, v_0)$ cannot contain only one bad color and be a subset of \opt.
Without loss of generality, assume $v_1$ is the vertex with a bad color in $W(T, v_0)$. Consider the time it is added to $W(T, v_0)$ by combination of some $W(Q, v_0) \setminus \{v_1, v_2\}$ and $W(\{\mathcal{C}(v_1)\}, v_1)$. However, our algorithm will choose to combine with $\bar{v_1}$ in stead of $v_1$ because $\bar{v_1}$ yields the largest increase in the normalized log-rank statistic.
To see this, note that $v_1$ and $\bar{v_1}$ both add two weights to $\sigma(W(Q, v_0))$ that are not already in $\sigma(W(Q, v_0))$. Both options therefore have the same number of mutations, and their normalized log-rank statistic can be compared by simply comparing their log-rank statistic. By construction $\sigma(\bar{v_0})$ contains the two largest weights, hence it yields the larger log-rank statistic.

\end{proof}

\begin{theorem}

For any optimal colorful connected subnetwork $\mathcal{S}$ of size $k \geq 3$ and any algorithm $\mathcal{A}$ which obtains subnetworks with colorsets of cardinality $i$ by combining $2$ subnetworks with colorsets of cardinality $< i$, by adding $3$ neighbors to $\mathcal{S}$ we have that $\mathcal{A}$ may not discover $S$ .
\end{theorem}
\begin{proof}[sketch]
Let the $k$ vertices of $OPT$ be deemed \textit{good} vertices. For each of three of the vertices in $OPT$ we add a \textit{bad} copy, so that the good vertex $v$ and the bad vertex $\bar{v}$ have the same color and the same connectivity to the vertices in $OPT\setminus \{v\}$.
By definition of $\mathcal{A}$, $\mathcal{S}$ is found by combining two subnetworks of cardinality $< k$, and because there are three good vertices in $OPT$, one of these subnetworks of cardinality $< k$ will contain at least two good vertices. We show that an evil adversary can ensure that two subnetworks $\mathcal{S}_1$ and $\mathcal{S}_2$, both being entries in $W$ and each containing a good vertex, will never be combined by $\mathcal{A}$.

The combination of $\mathcal{S}_1$ and $\mathcal{S}_2$ will happen across a specific edge in the graph between one vertex $v_1 \in\mathcal{S}_1$ and one vertex $v_2 \in\mathcal{S}_2$. If $v_2$ is a good vertex then there will be another subnetwork $\bar{\mathcal{S}_2}$ in $W$ with the same colorset as $\mathcal{S}_2$, namely in the column corresponding to the bad vertex $\bar{v_2}$, and since the connectivities of $v_2$ and $\bar{v_2}$ to $OPT$ are the same, $\mathcal{A}$ must select one of them. Due to the fact that $|\mathcal{S}_1 \cup \mathcal{S}_2| < k$ the adversary will be able to plant mutations so that $\bar{\mathcal{S}_2}$ is chosen over $\mathcal{S}_2$.
If $v_2$ is neither a good nor a bad vertex the same argument can be made to show that the adversary can ensure that $\mathcal{S}_2$ will not contain any good vertices.
\end{proof}

The following is a result that we need to prove the performance of \algname\ under the Planted Subnetwork Model.

\begin{proposition}
\label{prop:sum0}
For every censoring vector $c$: $\sum_{i=1}^m w_i = 0$.
\end{proposition}

\begin{proof}[Proof (Sketch)]
When $c_i=1$ for all $1 \le i \le m$, the we have
\begin{eqnarray*}
\sum_{i=1}^m w_i & = & \sum_{i=1}^m \left( c_i  - \sum_{j=1}^{i}\frac{c_j}{m-j+1}\right) \\
& = & \sum_{i=1}^m \left( 1  - \sum_{j=1}^{i}\frac{1}{m-j+1}\right) \\
& = & m - \sum_{i=1}^m \sum_{j=1}^i \frac{1}{m-j+1}\\
& = & m - \sum_{i=1}^m i \frac{1}{i}\\
& = & m - m \\
& = & 0.
\end{eqnarray*}

When one $c_i$ is switched to the value $0$, we have that the weight changes by a factor:
\begin{equation}
-1 + \sum_{j=i}^{m} \frac{1}{m-i+1} = 0
\end{equation}
where the $-1$ is subtracted to $w_i$, while the value $\frac{1}{m-i+1}$ is summed (i.e., not subtracted) to all terms $w_j$ with $j \ge i$. Therefore, any change to the censoring vector leaves $\sum_{i=1}^m w_i = 0$.
\end{proof}

Using the above, we can prove the following.

\begin{theorem}
Let $M$ be a mutation matrix corresponding to $m$ samples from the Planted Subnetwork Model. If $m \in \BOM{k^4(k+\varepsilon) \ln n}$ for a given constant $\varepsilon > 0$ and $\BO{\ln(1/\delta)e^k}$ color-coding iterations are performed, then our algorithm identifies the optimal solution $\Pathway$ to the \problemgraphname\ with probability $\ge 1 - \frac{1}{n^\varepsilon} - \delta$.
\end{theorem}

\begin{proof}[Proof (Sketch)]
Assume that $\Pathway$ is colorful.  We prove that if \algname\ has build a subnetwork (with $1 \le i < k$ vertices) consisting of vertices of $\Pathway$ only, then
if $m \in \BOM{k^2(k+\varepsilon) \ln n}$, \algname\ will expand such solution by only using vertices in $\Pathway$. Since \algname\ starts to build solutions from each 
vertex in $\Pathway$, this proves that \algname\ identifies the optimal solution. We show this by proving that \emph{any} set $\C \subset \G \setminus \Pathway$, when added
to \emph{any} subset $\Set \subset \Pathway$, does not provide an improvement in the score as just adding one of the genes in $\Pathway$.

From the properties of the Planted Subnetwork Model (PSM), we have that if $\Set$ is a subset of $\Pathway$, then $w(S) \ge \frac{c' m}{k}$, where $c'$ is a constant $>0$. 
For a set $\C \subset \G \setminus \Pathway$, we can consider it as a ``metagene'' that is mutated with a certain probability $q$ (constant) in each sample, where $q$
depends on the genes in $\C$.

From Property~\ref{prop:sum0}, we have that $\E[w(\Set \cup \C) - w(\Set) ]  = - q w(\Set) \le - q \frac{c' m}{k}$, since the sum of all weights $w_i$ is 0 and $\C$ adds weights
from a set of weights that must sum to $- w(\Set)$. From the properties of PSM, for a gene $g \in \Pathway \setminus \Set$ we have $w(\Set \cup \{g\}) - w(\Set) \ge \frac{c'' m}{k}$, with $c'' > 0$ constant.
Note that $w(\Set \cup \C) - w(\Set)$ is the sum of independent random variables, and each random variable can change the value of $w(\Set \cup \C) - w(\Set)$ by a value $<m$.
Moreover, the number of samples in which $\C$ can have mutations while $\Set$ does not is at least $\frac{m}{k}$ and at most $m$. We can therefore use Hoeffding inequality to bound the probability that $w(\Set \cup \C) > w( \Set \cup \{g\} )$ as follows:
\begin{eqnarray*}
\Pr( w(\Set \cup \C ) > w( \Set \cup \{g\} ) ) & = & \Pr( w(\Set \cup \C ) - w(\Set) > w( \Set \cup \{g\} ) - w(\Set) ) \\
& \le & e^{- d( (\frac{m}{k})^2 (\frac{m}{k})^2 / m^3 } \\
& \le & \frac{1}{n^{k+\varepsilon}}
\end{eqnarray*}
for an appropriate constant $d>0$ and for $m\in\BOM{k^4(k+\varepsilon) \ln n}$. By union bound on all sets $\C$ of cardinality $\le k$, we have that $\Pr( w(\Set \cup \C ) > w( \Set \cup \{g\} ) ) \le \frac{1}{n^{k+\varepsilon}} n^k = \frac{1}{n^\varepsilon}$. Therefore, when $m\in\BOM{k^4(k+\varepsilon) \ln n}$ and $\Pathway$ is colorful, then \algname\ finds $\Pathway$ with probability $\ge 1 - \frac{1}{n^\varepsilon}$. The probability that $\Pathway$ is not colorful in any of the $\BO{\ln(1/\delta)e^k}$ color-coding iterations is $\le \delta$. Therefore, by union bound the probability that \algname\ does not identify $\Pathway$ when $m\in\BOM{k^4(k+\varepsilon) \ln n}$ is $\le \delta + \frac{1}{n^{\varepsilon}}$, and the result follows.
\end{proof}

\section{Modifications to \algname}

We design two modifications of \algname\ that can solve some easy cases where \algname\ may not identify the highest scoring solution due to its subnetwork merging strategy:
 \begin{enumerate}[i)] \item we merge a subnetwork $W(T,u)$ not only with subnetworks $W(R,v)$ where $v$ is a neighbor of $u$, but with subnetworks $W(R,w)$ where $w$ is a neighbor of \emph{any} vertex in $W(T,u)$; \item in $W(T,u)$, we store $\ell > 1$ different colorful subnetworks containing $u$ and with colorset $T$, leading to $\le \ell^2$ choices for combining two entries of $W$ and a corresponding $\ell^2$ increase in the time complexity of the algorithm.
\end{enumerate}
We note that the time complexity required by modification i) above is still polynomial at most a factor $|V|^2/|E| \in \BOM{n}$ larger than that of \algname. We note that both modifications find the optimal solution in the problem instance of Proposition~\ref{thm:badfamily}, while the second one will find the optimal solution in the problem instance of Theorem~\ref{thm:colorbad} if $\ell$ is large enough. 
The second modification was run using $\ell=5$ in our experiments and storing in $W(T,u)$ the $\leq \ell$ highest scoring subnetworks in $\mathcal{S}^\prime(T,u)$.

\section{Greedy algorithms}

We considered three different greedy strategies for the \problemgraphname\ problem. All three algorithms build solutions starting from each node $u \in G$ in iterations by adding nodes to the current solution $\mathcal{S}$, and differ in the way the enlarge the current subnetwork $\mathcal{S}$ of size $1\leq i < k$.
The first, \naive, screens all vertices at distance $1$ to $\mathcal{S}$ and adds the one that results in the best subnetwork of size $i+1$. The second, \fss, considers all vertices at distance $\leq k-i$ to $\mathcal{S}$, and enforces connectivity by greedily constructing a path from the selected vertex to a vertex in $\mathcal{S}$. The third, \dfs, traverses shortest paths from $\mathcal{S}$ to every vertex at distance $\leq k-i$ by a depth-first search. The vertices on some shortest path of length $j \leq k-i$ which improved $\mathcal{S}$ the most are added to obtain a subnetwork of size $i+j$.

\newpage
\section{Pseudo code for NoMAS}
The pseudo code for \algname\ is divided into three algorithms. First, algorithm~\ref{alg:nomas} highlights the overall color-coding scheme. Second, algorithm~\ref{alg:fillTable} describes how the dynamic programming table $W$ is computed in order of increasing colorset group sizes. Finally, algorithm~\ref{alg:computeEntry} details the process of computing the subnetwork at a specific entry in $W$. It is assumed that the undirected graph $G(V,E)$, the mutation matrix $M$ and the survival information $\mathbf{x}, \mathbf{c}$ are globally known. As a companion piece to algorithm~\ref{alg:computeEntry}, figure~\ref{fig:computeEntryNormal} visualizes the method used for combining two previously computed entries of $W$.

\begin{algorithm}[H]
\SetAlgoLined
\label{alg:nomas}
$\Set \leftarrow \emptyset$\\
\For{\upshape $i\leftarrow 1$ \textbf{to} $ln(\frac{1}{\delta})e^k$} {
	Color the vertices of $G$ with $k$ colors uniformly at random\\
	$W \leftarrow$ \textsc{FillTable}($k$)\\
	$\Set^\prime \leftarrow \displaystyle\argmax_{\forall T \forall v~:~ W(T,v) \in W}\left\{w(W(T,v))\right\}$\\
	$\Set \leftarrow \displaystyle\argmax \left\{w(\Set), w(\Set^\prime)\right\}$ 
}
\Return $\Set$
\caption{\textsc{NoMAS}($k$, $\delta$)}
\end{algorithm}

\begin{algorithm}[H]
\SetAlgoLined
\label{alg:fillTable}
$W \leftarrow$ empty table with dimensions $(2^k-1) \times |V|$\\
\For{\upshape\textbf{each} vertex $u\in V$} {
	\For{\upshape\textbf{each} color $\alpha$ among the $k$ colors} {
		\If{\upshape the color of $u$ is $\alpha$} {
			$W(\{\alpha\},u) \leftarrow \{u\}$
		}
		\Else{
			$W(\{\alpha\}, u) \leftarrow \emptyset$
		}
	}
}
\For{\upshape$i\leftarrow 2$ \textbf{to} $k$} {
	\tcc*[h]{\footnotesize The following may be distributed among $N \leq |V|$ processors}\\
	\For{\upshape\textbf{each} vertex $u\in V$} {
		\For{\upshape\textbf{each} colorset $T$ of size $i$} {
			$W(T,u) \leftarrow$ \textsc{ComputeEntry}($T$, $u$)
		}
	}
}
\Return $W$
\caption{\textsc{FillTable}($k$)}
\end{algorithm}

\begin{algorithm}[H]
\SetAlgoLined
\label{alg:computeEntry}
best $\leftarrow \emptyset$\\
\For{\upshape\textbf{each} neighbor $v$ of $u$} {
	\For{\upshape\textbf{each} colorset $Q$ s.t. $Q \subset T$ and $Q \neq \emptyset$} {
		$R \leftarrow T\setminus Q$\\
		\If{\upshape $W(Q,u) \neq \emptyset$ \textbf{and} $W(R,v) \neq \emptyset$} {
			candidate $\leftarrow W(Q, u) \cup W(R, v)$\\
			best $\leftarrow \argmax\{w($candidate$), w($best$)\}$
		}
	}
}
\Return best
\caption{\textsc{ComputeEntry}($T$, $u$)}
\end{algorithm}

\begin{figure}[h!]
\centering
\includegraphics[width=0.99\textwidth]{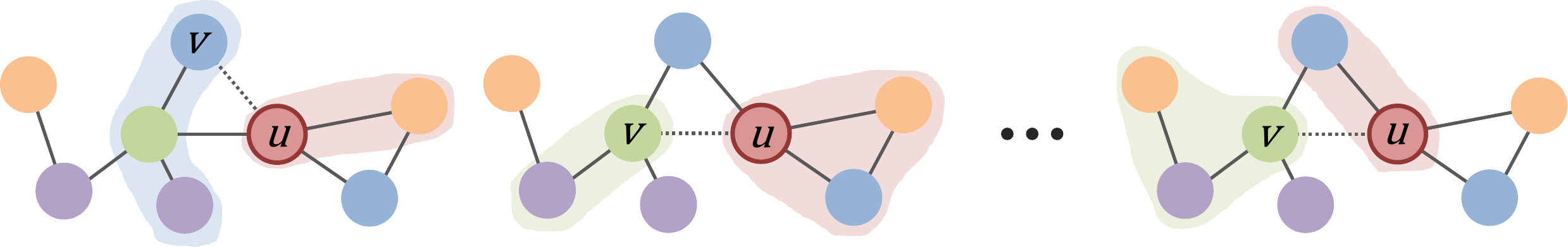}
\caption{Examples of several pairs of colorful connected subnetworks $W(Q, u)$ and $W(R, v)$ considered by \algname\ when computing the entry $W(T,u)$ for a colorset $T$ of size $5$. In each example a subnetwork containing $u$ are combined with a subnetwork containing a neighbor $v$ of $u$, in order to obtain a subnetwork with colorset $T = Q \cup R$ such that $Q \cap R = \emptyset$. The dotted edge is the one connecting the two subnetworks (the edge is always connected to $u$).}
\label{fig:computeEntryNormal}
\end{figure}

\paragraph{Modifications} The two proposed modifications to \algname\ differ from \algname\ in their method for computing an entry of $W$. Algorithm~\ref{alg:computeEntryModified1} describes modification i, while algorithm~\ref{alg:computeEntryModified2} details modification ii. Both algorithms should be seen as replacements for algorithm~\ref{alg:computeEntry} of the unmodified version of \algname. Figure~\ref{alg:computeEntryModified2} visualizes the combination strategy of algorithm~\ref{alg:computeEntryModified1} (note the difference from figure~\ref{fig:computeEntryNormal}).

\begin{algorithm}[H]
\SetAlgoLined
\label{alg:computeEntryModified1}
best $\leftarrow \emptyset$\\
\For{\upshape\textbf{each} colorset $Q$ s.t. $Q \subset T$ and $Q \neq \emptyset$} {
		$R \leftarrow T\setminus Q$\\		
		\For{\upshape\textbf{each} neighbor $w$ of a vertex in $W(Q,u)$} {
		candidate $\leftarrow W(Q, u) \cup W(R, w)$\\
		best $\leftarrow \argmax\{w($candidate$), w($best$)\}$
	}
}
\Return best
\caption{\textsc{ModificationI}($T$, $u$)}
\end{algorithm}

\begin{algorithm}[H]
\SetAlgoLined
\label{alg:computeEntryModified2}
candidates $\leftarrow\emptyset$\\
\For{\upshape\textbf{each} neighbor $v$ of $u$} {
	\For{\upshape\textbf{each} colorset $Q$ s.t. $Q \subset T$ and $Q \neq \emptyset$} {
		$R \leftarrow T\setminus Q$\\		
		\For{\upshape\textbf{each} subnetwork $A\in W(Q,v)$} {
			\For{\upshape\textbf{each} subnetwork $B\in W(R,v)$} {
				candidates $\leftarrow$ candidates $\cup \{A \cup B\}$\\
			}
		}
	}
}
best $\leftarrow$ the $\ell$ distinct highest scoring subnetworks in candidates\\
\Return best
\caption{\textsc{ModificationII}($T$, $u$)}
\end{algorithm}

\begin{figure}[h!]
\centering
\includegraphics[width=0.99\textwidth]{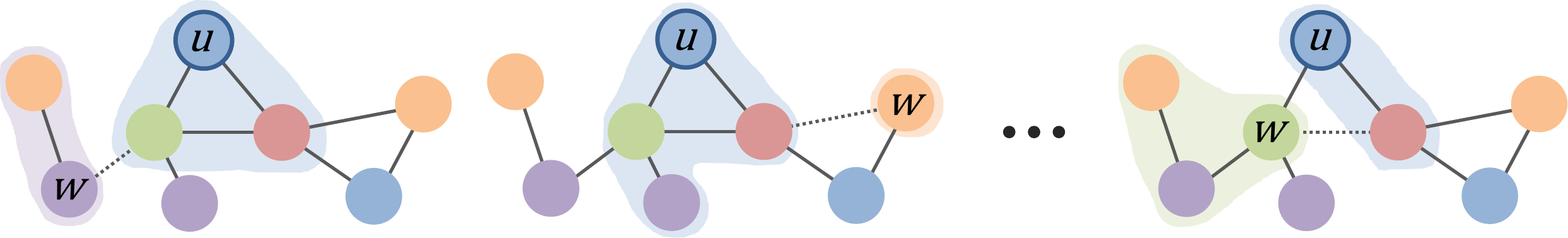}
\caption{Examples of several pairs of colorful connected subnetworks $W(Q, u)$ and $W(R, w)$ considered by \algname\ with modification i when computing the entry $W(T,u)$ for a colorset $T$ of size $5$. In each example a subnetwork $W(Q,u)$ containing $u$ are combined with a subnetwork containing a neighbor $w$ of some vertex in $W(Q,u)$, in order to obtain a subnetwork with colorset $T = Q \cup R$ such that $Q \cap R = \emptyset$. The dotted edge is the one connecting the two subnetworks.}
\label{fig:computeEntryModified}
\end{figure}

\newpage

\begin{suppfigure}[ht]
\centering
\includegraphics[scale=0.4]{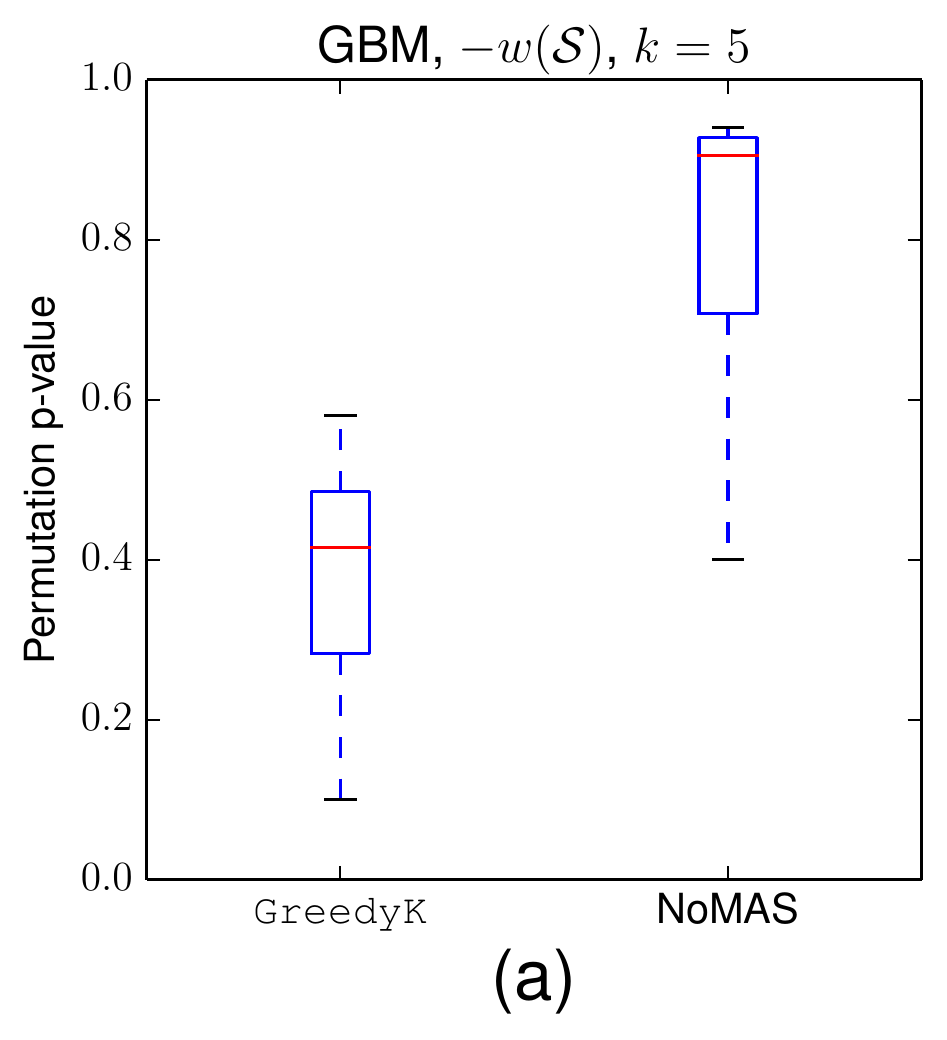}
\includegraphics[scale=0.4]{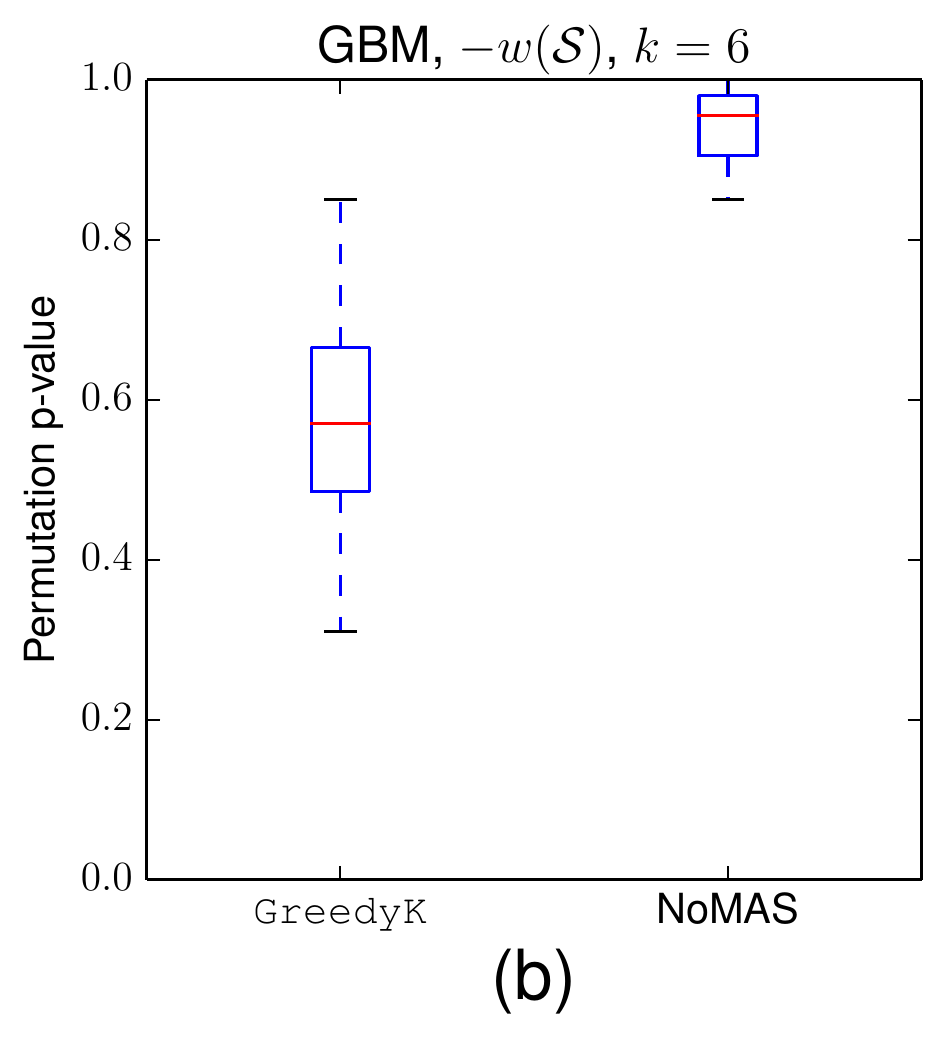}
\includegraphics[scale=0.4]{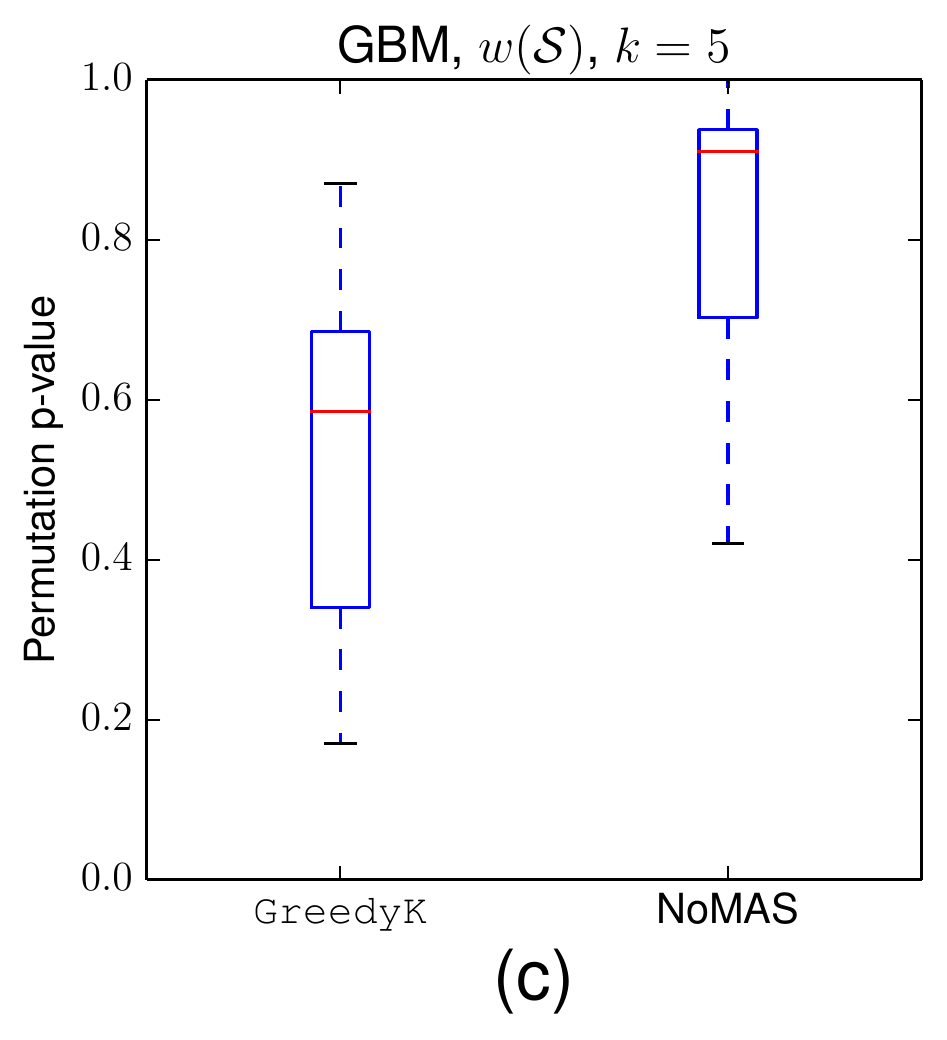}
\includegraphics[scale=0.4]{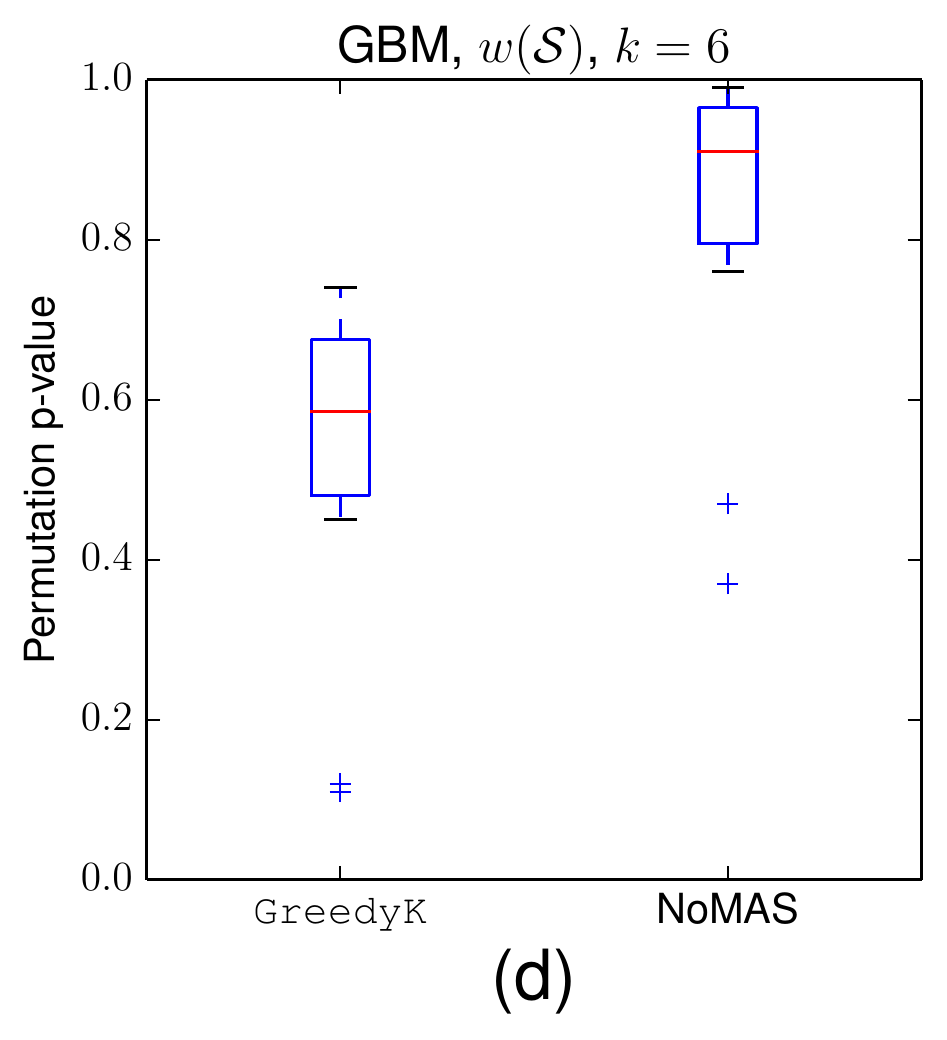}
\caption{p-values of the permutation test on the top $10$ solutions identified by \fss\ for different values of $k$ in both tail tests. The top $10$ solutions on the permuted data are obtained using both \fss\ and \algname\ (with $32$ color-coding iterations).}
\label{fig:GBMpval}
\end{suppfigure}

\begin{suppfigure}[ht]
\centering
\includegraphics[scale=0.45]{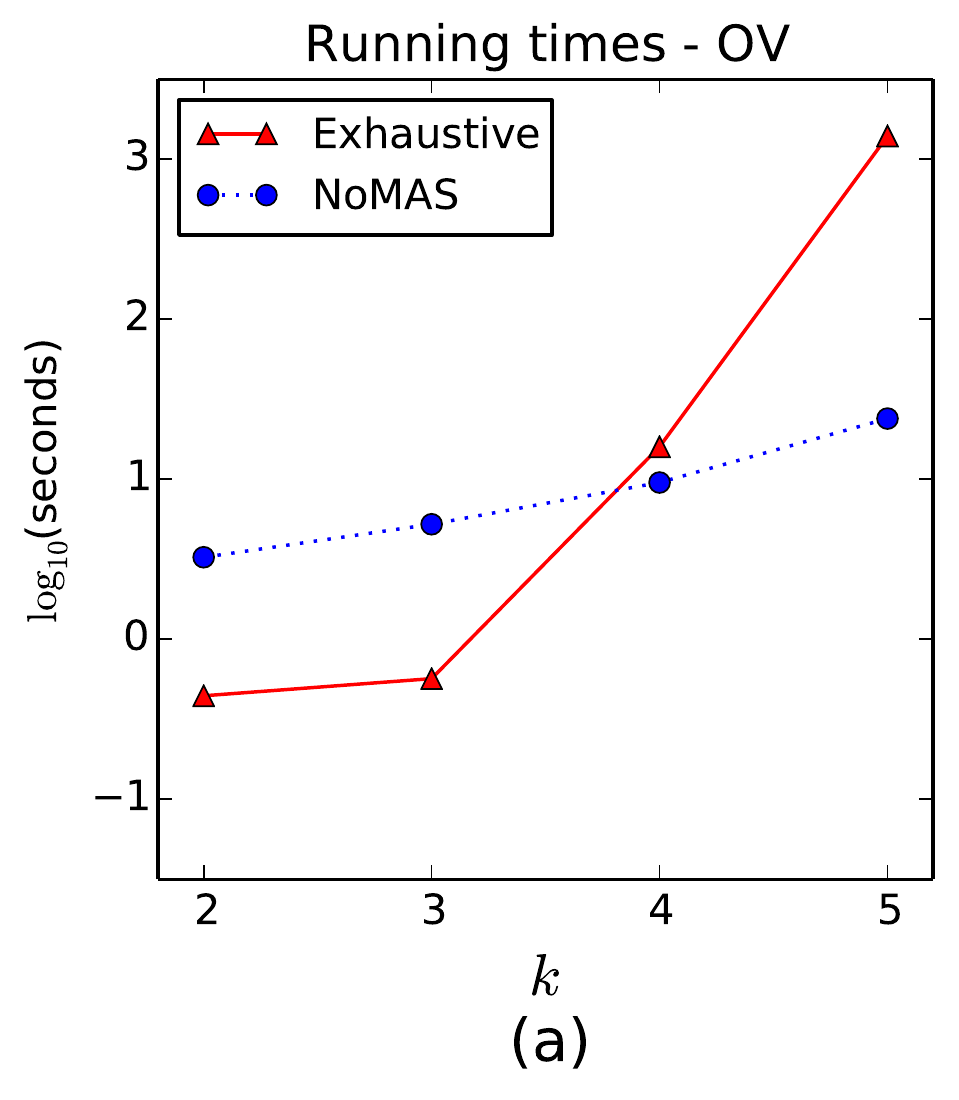}
\includegraphics[scale=0.45]{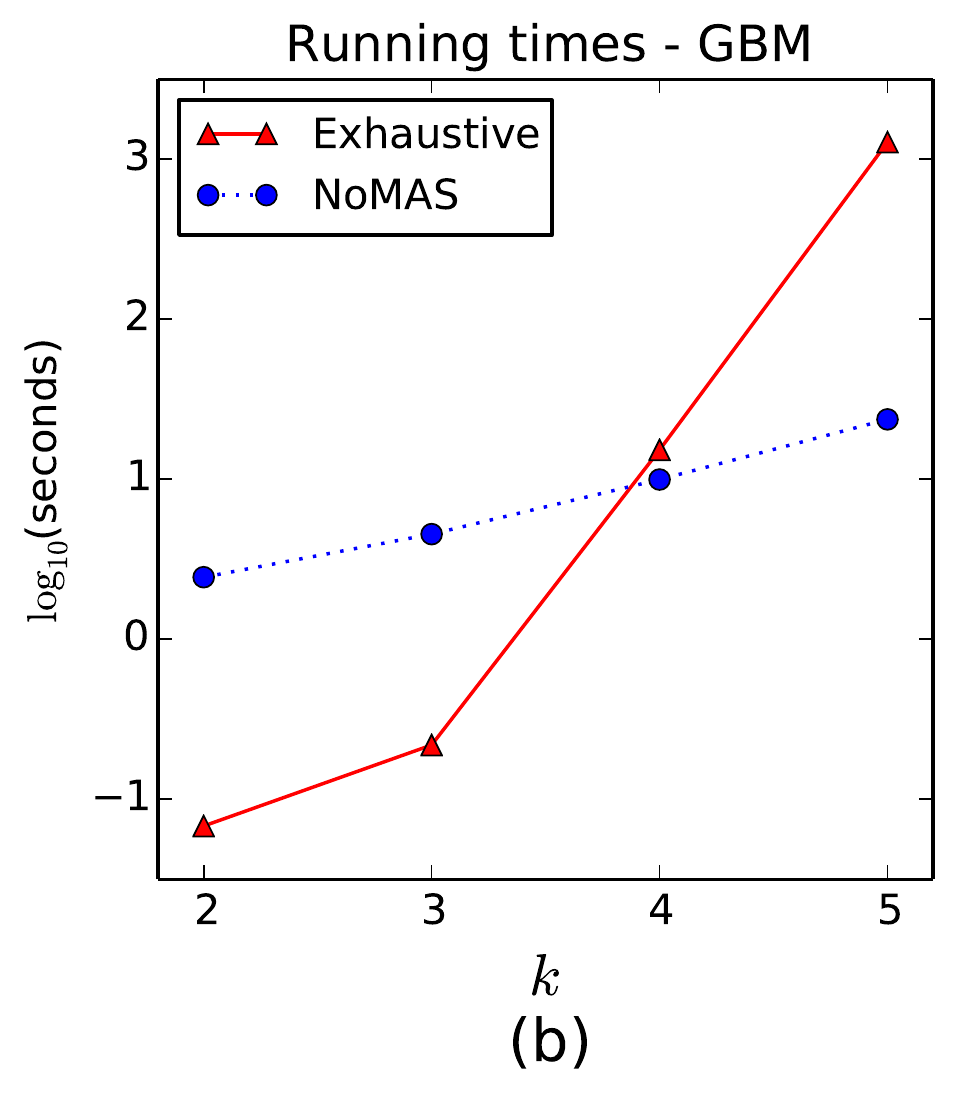}
\includegraphics[scale=0.45]{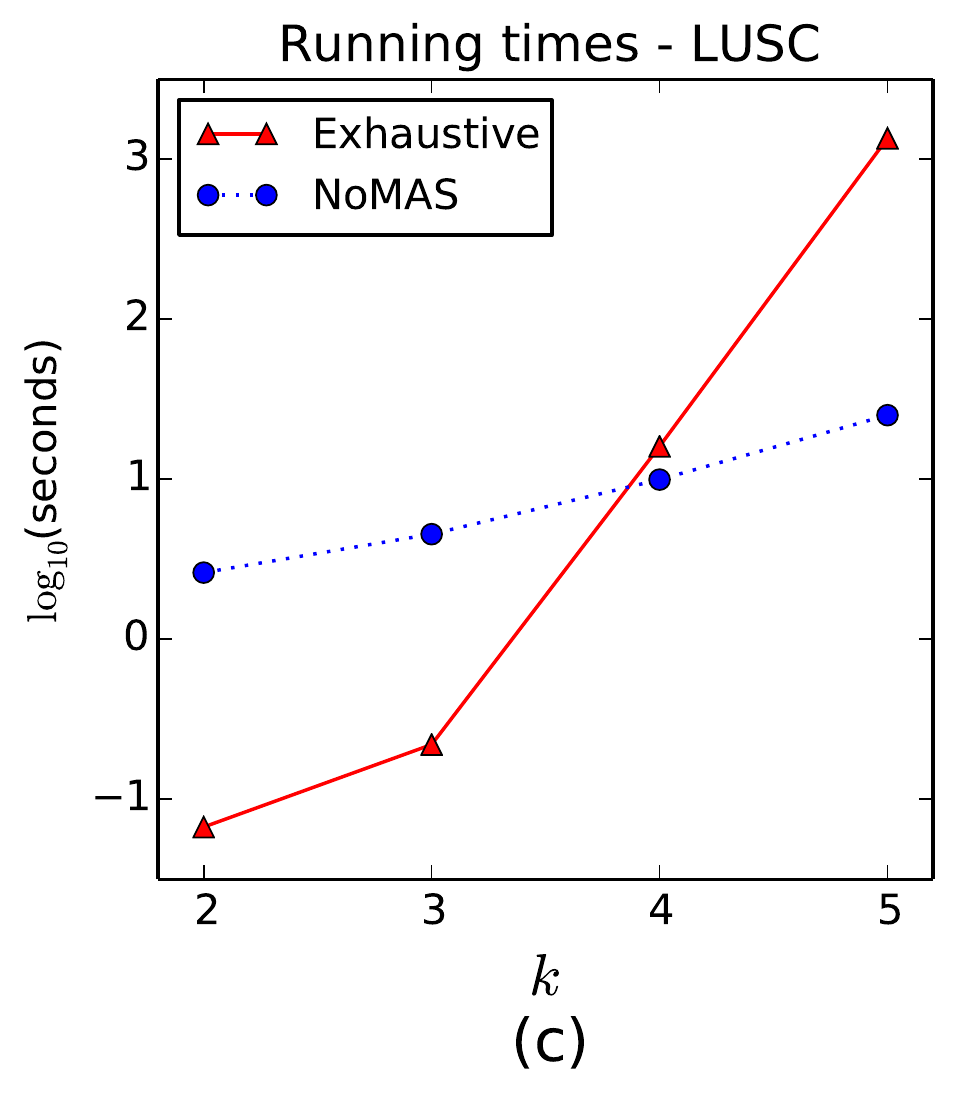}
\caption{Running time comparison between \algname\ and the exhaustive enumeration algorithm on three different cancer datasets. The running times of both algorithms are obtained using $40$ processors. The running times for \algname\ account for $256$ color-coding iterations and excludes the statistical assessment of the identified solutions.}
\label{fig:runtime}
\end{suppfigure}

\begin{suppfigure}[ht]
\centering
\includegraphics[scale=0.45]{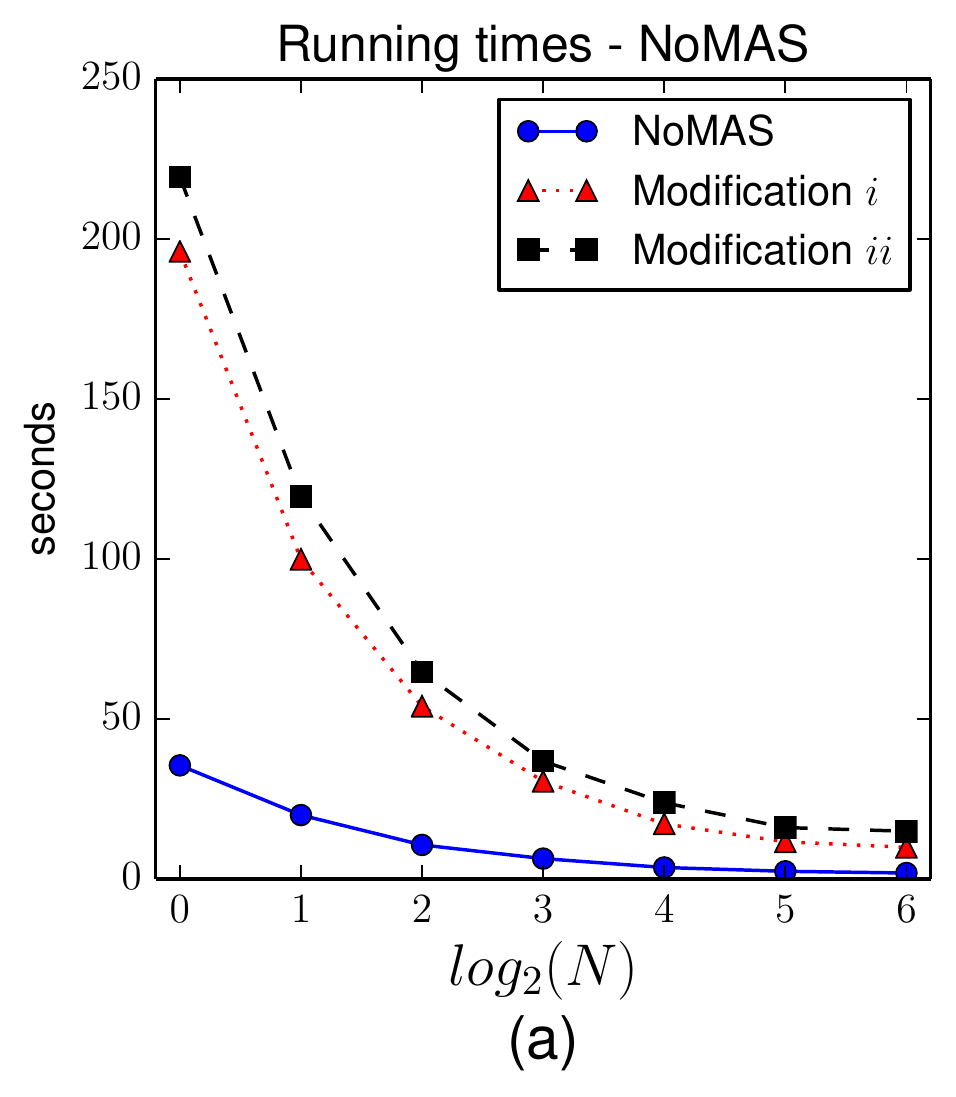}
\includegraphics[scale=0.45]{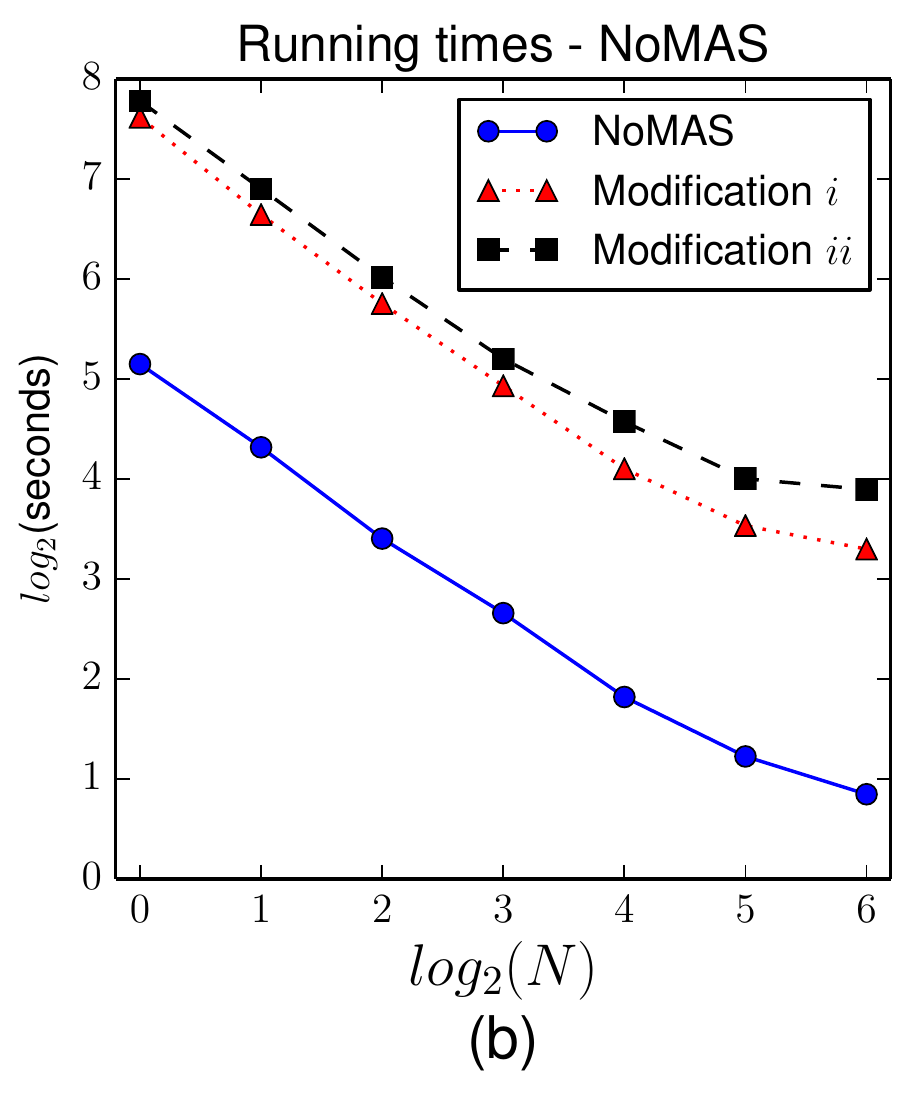}
\caption{Running times of \algname\ and the two modifications considered for varying numbers of processors $N$. The running times are for a single iteration for $k=8$ and are obtained on the OV cancer data (a) The running times in seconds. (b) The running times in seconds on a logarithmic scale.}
\label{fig:parallel}
\end{suppfigure}

\begin{suppfigure}[ht]
\centering
\includegraphics[scale=0.5]{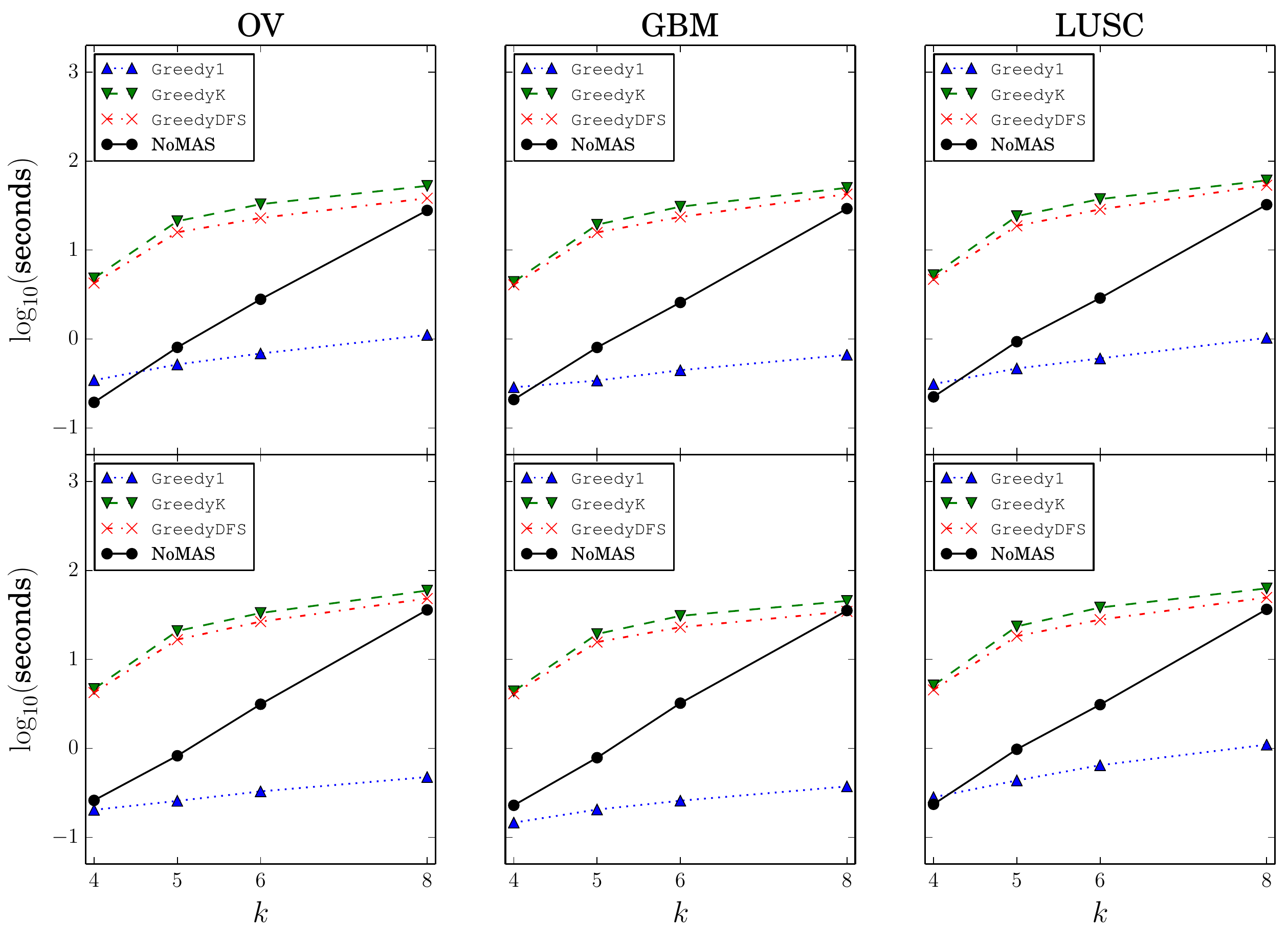}
\caption{Running times of the three greedy algorithms and a single color-coding iteration of \algname\ for varying values of $k$ and on three different cancer data. Each of the algorithms are run on a single processor. The top panels show the times measured when maximizing the score $w(\mathcal{S})$, while the bottom panels show the times for maximizing the score $-w(\mathcal{S})$.}
\label{fig:nomasGreedyTimes}
\end{suppfigure}

\end{document}